\documentclass{IEEEtran}

\input{preamble}

\begin{document}

\title{Integrating Planning and Predictive Control\\Using the Path Feasibility Governor}
\author{%
    Shu~Zhang,
    James~Y.~Z.~Liu,
    and Dominic~Liao-McPherson%
    \thanks{This work was supported by the Natural Sciences and Engineering Research Council of Canada Discovery Grants Program (Reference \#\,RGPIN-2023-03257) and the National Research Council Collaborative Science, Technology, and Innovation Program (Agreement \#\,CSDP-016-1)}%
    \thanks{Shu Zhang (e-mail: shu.zhang@ubc.ca), James Y. Z. Liu and Dominic Liao-McPherson (e-mail: dliaomcp@mech.ubc.ca) are with the University of British Columbia, Vancouver, BC, Canada.}%
}
\maketitle

\begin{abstract}
    The motion planning problem of generating dynamically feasible, collision-free trajectories in non-convex environments is a fundamental challenge for autonomous systems. Decomposing the problem into path planning and path tracking improves tractability, but integrating these components in a theoretically sound and computationally efficient manner is challenging. We propose the Path Feasibility Governor (PathFG), a framework for integrating path planners with nonlinear Model Predictive Control (MPC). The PathFG manipulates the reference passed to the MPC controller, guiding it along a path while ensuring constraint satisfaction, stability, and recursive feasibility. The PathFG is modular, compatible with replanning, and improves computational efficiency and reliability by reducing the need for long prediction horizons. We prove safety and asymptotic stability with a significantly expanded region of attraction, and validate its real-time performance through simulated case studies of unicycle and quadrotor navigation in cluttered environments.
\end{abstract}

\section{Introduction}

\IEEEPARstart{M}{otion} planning involves generating collision-free and dynamically feasible trajectories that steer a system from an initial state to a target state, often under complex, non-convex constraints. This problem is of fundamental importance across nearly all areas of engineering, e.g., aerial infrastructure inspection \cite{ref:Shi2021}, robot-assisted surgery \cite{ref:Baek2018}, automated logistics systems \cite{ref:Wang2015}, and underwater exploration \cite{ref:An2022}.

Motion planning is intrinsically challenging due to the high-dimensional and non-convex nature of the search space \cite{ref:Hwang1992}. Direct trajectory optimization can solve motion planning problems, but often suffers from prohibitive computational complexity and is prone to falling into local minima \cite{ref:Sandakalum2022}. To improve computational tractability and robustness to disturbances, the motion planning problem is commonly decomposed into two hierarchical subproblems: path planning and path tracking. Path planners (e.g., RRT, A\!*, and Potential Field) generate collision-free geometric paths \cite{ref:Costa2019,ref:Karur2021}, while path tracking controllers (e.g., PID, LQR, and MPC) are feedback controllers that drive the dynamic system along the path \cite{ref:Brogan1991,ref:Goodwin2006}. This decomposition reduces computational complexity but introduces a kinodynamic disconnect: path planners often ignore or approximate system dynamics, producing kinematically feasible but dynamically infeasible paths. Integrating planning and control in a way that is both computationally tractable and guarantees that the closed-loop system can track the path and satisfy constraints remains a core challenge in motion planning.

\begin{figure}
    \centering
    \includegraphics[width=\columnwidth, trim={0.6cm 0.2cm 1.1cm 0}, clip]{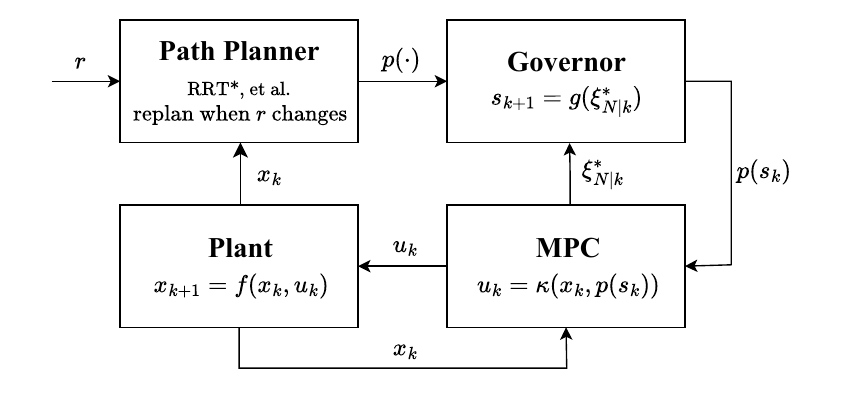}
    \caption{Block diagram of the control architecture. Given a desired target reference $r$, the path planner generates a path connecting the initial position to the target. The Path Feasibility Governor selects from this path an auxiliary reference $s$ that is guaranteed to be feasible for MPC to track, thereby ensuring MPC can compute a valid control input $u$.}
    \label{fig:ctrlArchitecture}
\end{figure}

A broad class of heuristic and ad hoc methods has been developed to bridge path planning and path tracking. Approaches like those in \cite{ref:Kalakrishnan2011,ref:Williams2017,ref:Zhou2019,ref:Ryll2019,ref:Tordesillas2020,ref:Wallace2024} generate dynamically feasible trajectories by sampling and optimizing multiple candidate paths, while the Safe Flight Corridor (SFC) method \cite{ref:Liu2017,ref:Gao2018,ref:Wu2021} constructs convex polyhedral safety margins around reference paths to facilitate trajectory optimization. Other techniques \cite{ref:Liu2024,ref:Zhang2022,ref:Al-Moadhen2022} first interpolate paths at empirically determined resolutions before feeding waypoints to tracking controllers. This class of methods performs well in practice but can fail unexpectedly due to the lack of theoretical guarantees.

Another major category of approaches is controller-agnostic methods, which modify path planning algorithms to ensure trackability but are agnostic to the exact controller. Hamilton-Jacobi (HJ) reachability-based frameworks, such as FaSTrack \cite{ref:Herbert2017,ref:Chen2021,ref:Gong2024}, use HJ reachability analysis to compute a Tracking Error Bound (TEB) that quantifies the maximum deviation between the planning and dynamic models, and incorporate this bound into the planning algorithm to ensure the path is dynamically feasible. This approach is theoretically sound but can be prohibitively computationally expensive due to the curse-of-dimensionality challenges inherent in HJ methods.

In contrast, controller-integrated methodologies consider specific controllers. Kinodynamic RRTs \cite{ref:LaValle1999,ref:LaValle2001} directly integrate system dynamics into RRT algorithms but are often computationally expensive. LQR-trees \cite{ref:Tedrake2010} integrate planning and control by verifying the regions of attraction of LQR feedback policies. Invariant-set planners \cite{ref:Weiss2015,ref:Danielson2016,ref:Weiss2017,ref:Berntorp2017,ref:Danielson2020} formulate path planning as a graph search, augmenting nodes with controllers (e.g., LQR) and positive-invariant subsets. Edges between nodes represent safe transitions across overlapping invariant sets, guaranteeing safety and dynamical feasibility. The recent work \cite{ref:Csomay-Shanklin2024} similarly adopts a graph-based framework, using dynamically feasible Bézier curves to connect nodes. This family of methodologies properly integrates path planning and tracking; however, incorporating the controller into the planning process make it difficult to integrate complex/performant controllers such as MPC \cite{ref:Diehl2018,ref:Goodwin2006,ref:Mayne2014}.

In this work, we propose the Path Feasibility Governor (PathFG), a modular system that enables the integration of predictive controllers with a wide array of path planners. The proposed architecture is illustrated in Figure~\ref{fig:ctrlArchitecture}. Given a reference $r$, the path planner (e.g., RRT* or Potential Field) generates a collision-free path. The PathFG then creates a sequence of intermediate references that guides the system as quickly as possible along the planned path, while ensuring the intermediate references are feasible for the predictive controller. This architecture offers several key advantages. The PathFG is computationally cheap and ensures constraint satisfaction, asymptotic stability, and recursive feasibility of the MPC. It also enlarges the region of attraction (ROA) of the closed-loop system to any state that can be connected to the reference by the planner reducing the need for long prediction horizons. Furthermore, the PathFG decouples non-convex obstacle avoidance from the simpler tracking task, enabling the MPC optimization problem to be solved more efficiently and reliably. In addition, its modular design allows integration with any path planning algorithm that generates feasible paths, owing to the clear separation between planning and tracking. This contrasts with methods such as kinodynamic RRTs, LQR-trees, invariant-set planners, etc., that embed the controller into the planning process. The modular structure also avoids the need to modify the MPC controller, allowing the PathFG to be readily integrated with existing MPC toolkits.

The PathFG is an extension of the feasibility governors (FGs) \cite{ref:Skibik2021,ref:Skibik2022,ref:Skibik2023} from convex to non-convex search spaces and significantly broadens the set of applications that can be tackled with the method. The explicit reference governor (ERG) \cite{ref:Nicotra2018,ref:Convens2021,ref:Li2023} addresses a problem similar to that of the PathFG in continuous-time settings. Several efforts have been made to discretize the ERG for use in discrete-time systems \cite{Convens2024,Merckaert2024}. However, these approaches require additional supervisory logic and careful parameter tuning to ensure that the discretization preserves both feasibility guarantees and satisfactory performance. Another closely related line of research (see \cite{ref:krupa2024} for a survey) incorporates an intermediate artificial reference into the MPC optimization problem itself. This method (MPCa) integrates planning and control into a single module, enabling online reference changes without recomputing a feasible path. However, in ``highly'' non-convex settings commonly encountered in practice, MPCa may face challenges, as gradient-based solvers typically used can become trapped in local minima. While \cite{Santos2023,Santos2024,Nubert2020} assume the global solution can be attained, this is often unrealistic in real-world scenarios. The approach in \cite{Cotorruelo2021} addresses this by convexifying certain problems, but is limited to domains without obstacles. In such cases, the PathFG framework, designed primarily for obstacle avoidance, offers significant advantages.

The remainder of this paper is organized as follows. The problem formulation is presented in Section~\ref{sec:problemFormulation}, followed by the controller design in Section~\ref{sec:controllerDesign}. Section~\ref{sec:theory} formally establishes the PathFG's safety and asymptotic stability properties. Section~\ref{sec:Unicycle} presents the implementation of PathFG in a nonlinear unicycle navigation example, while Section~\ref{sec:implementation} details its implementation for a class of systems and demonstrates its performance through simulations of quadrotor navigation in an obstacle-cluttered environment.

\vspace{\thmsep}
\noindent\emph{Notation}: $\mathbb R^{N \times M}\!$ denotes the set of real $N \times M$ matrices, and $\mathbb R^N\!$ denotes the set of real $N$-dimensional vectors. The $N \times N$ identity matrix is denoted by $I_N$, and the $N \times M$ zero matrix is denoted by $0_{N \times M}$. The subscripts for the identity and zero matrices are omitted whenever the dimensions are clear from the context. The relation ``$\defas$'' is used for definitions, e.g., ``$a \defas b\+$'' means that $a$ is defined as $b$. Positive definiteness and positive semi-definiteness of a matrix $A$ are denoted by $A \succ 0$ and $A \succeq 0$, respectively. $\mathbb N_{[a,b]} \defas \set{n \in \mathbb N | a \leq n \leq b}$ denotes the set of natural numbers between $a$ and $b$. In this article, we use the definition of $\mathbb N$ that includes 0. The closed ball of radius $\epsilon \in \mathbb R$ around the point $p \in \mathbb R^{n_p}$ is denoted by $\mathcal B_\epsilon(p) \defas \set{p' \mid \norm{p' - p} \leq \epsilon}$. $\norm{x}$ denotes the Euclidean norm of $x$. The transpose of a vector $a$ is $a^\T$. The concatenation of two vectors $a$ and $b$ is denoted as $(a, b) \defas [a^\T\;b^\T]^\T$. The sets $\mathbb R_{>0}$, $\mathbb R_{\geq0}$ and $\mathbb N_{>0}$ denote the set of positive real numbers, non-negative real numbers and positive integers, respectively. The complement of a set $\mathcal S$ is denoted as $\mathcal S^\complement$, and the closure is $\closure{\mathcal S}$. $\limsup$ denotes the limit superior.

\section{Problem Formulation} \label{sec:problemFormulation}

Consider the nonlinear time-invariant system
\begin{equation} \label{eq:nonlinearsysEqs}
    x_{k+1} = f(x_k, u_k)
\end{equation}
subject to the pointwise-in-time constraints
\begin{equation} \label{eq:constr} \begin{gathered}
    x_k \in \mathcal X \defas \set{x | h_x(x) \leq 0}\\
    u_k \in \mathcal U \defas \set{u | h_u(u) \leq 0}\mathrlap,
\end{gathered}\end{equation}
where $k \in \mathbb N$ is the discrete-time index, $x_k \in \mathbb R^{n_x}$ is the state vector, $u_k \in \mathbb R^{n_u}$ is the control input, $\mathcal X$ is the state constraint set, $\mathcal U$ is the control input constraint set, $f\colon \mathbb R^{n_x} \times \mathbb R^{n_u} \to \mathbb R^{n_x}$, $h_x\colon \mathbb R^{n_x} \to \mathbb R$, and $h_u\colon \mathbb R^{n_u} \to \mathbb R$.

The following assumptions ensure that the system \ref{eq:nonlinearsysEqs} admits a well-posed tracking problem.

\begin{assume} \label{as:system}
    The function $f$ is locally Lipschitz, and the functions $h_x$ and $h_u$ are continuous. There exists a set of steady-state admissible references $\mathcal R \subseteq \mathbb R^{n_r}$ and continuous functions $\bar x_{(\cdot)}\colon \mathcal R \to \mathcal X$ and $\bar u_{(\cdot)}\colon \mathcal R \to \mathcal U$ such that for all $r \in \mathcal R$, $\bar x_r = f(\bar x_r, \bar u_r)$.
\end{assume}

We also define the set of strictly steady-state admissible references as
\begin{equation} \label{eq:admissibleRef}
    \mathcal R_\epsilon \defas \set{r \in \mathcal R | h_x(\bar x_r) \leq -\epsilon,\, h_u(\bar u_r) \leq -\epsilon}
\end{equation}
for some $\epsilon \in \mathbb R_{>0}$.

We are now ready to formally state the control problem addressed by this paper.
 
\vspace{\thmsep}
\noindent\emph{Control objectives}: given the system \ref{eq:nonlinearsysEqs} subject to constraints \ref{eq:constr}, let $r \in \mathcal R_\epsilon$ be a target reference. Our objective is to design a compensator 
\begin{equation}\begin{aligned}
    s_{k+1} &= g(x_k, s_k, r)\\
    u_k &= \kappa(x_k, s_k, r)
\end{aligned}\end{equation}
with internal state $s_k \in \mathbb R^{n_s}$, $g\colon \mathbb R^{n_x} \times \mathbb R^{n_s} \times \mathbb R^{n_r} \to \mathbb R^{n_s}$ and $\kappa\colon \mathbb R^{n_x} \times \mathbb R^{n_s} \times \mathbb R^{n_r} \to \mathcal U$ such that the closed-loop system
\begin{equation}\begin{aligned} \label{eq:closed-loop-system}
    x_{k+1} &= f(x_k, u_k)\\
    s_{k+1} &= g(x_k, s_k, r)\\
    u_k &= \kappa(x_k, s_k, r)
\end{aligned}\end{equation}
satisfies
\begin{enumerate}
    \item Safety: $x_k \in \mathcal X$ and $u_k \in \mathcal U$ for all $k \in \mathbb N$;
    \item Asymptotic stability: $\bar x_r$ is a asymptotically stable equilibrium point of the closed-loop system. 
\end{enumerate}
For clarity of presentation, we don't consider disturbances explicitly; the PathFG is an add-on unit that inherits its robustness properties from the underlying predictive controller.

\section{Controller Design} \label{sec:controllerDesign}

For this constrained control problem, we propose the PathFG unit to integrate a path planner and MPC controller. The idea behind the PathFG is to select intermediate references along the path that are always reachable within the prediction horizon of the MPC controller, thereby ensuring that the optimal control problem (OCP) underlying the MPC feedback law remains feasible at all times. The control architecture is illustrated in Figure~\ref{fig:ctrlArchitecture}. Given a desired target reference $r$, the path planner generates a path connecting the initial position to the target. The PathFG uses this path to create a sequence of intermediate references that are then passed to the MPC. 

\subsection{Nonlinear MPC}

Since the control problem is constrained, we approach the control objectives using a standard nonlinear MPC formulation, where the feedback policy is determined by solving the following OCP:
\begin{subequations} \label{eq:OCP} \begin{gather}
    \zeta^*(x, r) \defas \argmin_{\xi,\mu}\qty\bigg[V(\xi_N, r) +\! \sum_{i=0}^{N-1} \ell(\xi_i, \mu_i, r)] \label{eq:cost}\\
    \begin{aligned} \label{eq:MPCConstr}
        \text{s.t.} \quad & \xi_0 = x\\
            & \xi_{i+1} = f(\xi_i, \mu_i), && \forall i \in \mathbb N_{[0,N-1]}\\
            & \xi_i \in \mathcal X,\, \mu_i \in \mathcal U, && \forall i \in \mathbb N_{[0,N-1]}\\
            & (\xi_N, r) \in \termset,
    \end{aligned}
\end{gather}\end{subequations}
where $N \in \mathbb N_{>0}$ is the length of the prediction horizon; $\mu \defas (\mu_0,\allowbreak \mu_1, \dots, \mu_{N-1})$ are the control inputs; $\xi \defas (\xi_0, \xi_1, \dots, \xi_N)$ are the predicted states; $\ell\colon \mathcal X \times \mathcal U \times \mathcal R_\epsilon \to \mathbb R_{\geq0}$ is the stage cost; $V\colon \termset \to \mathbb R_{\geq0}$ is the terminal cost; and $\termset \subseteq \mathcal X \times \mathcal R$ is the terminal set.

A solution exists only if the terminal set $\termset$ associated with the target reference $r$ is reachable in $N$ steps from the current state $x$. The feasible set is the set of $(x, r)$ pairs for which the OCP \ref{eq:OCP} admits a solution, i.e.,
\begin{equation}
    \Gamma \defas \set{(x, r) \in \mathcal X \times \mathcal R_\epsilon | \exists (\xi, \mu) \text{ s.t. \ref{eq:MPCConstr}}}.
\end{equation}
This feasible set $\Gamma$ is the $N$-step backward reachable set of $\termset$.

The solution mapping for \ref{eq:OCP} can be defined as a set-valued function $\bar\zeta^*\colon \Gamma \rightrightarrows \mathcal X^{N+1} \!\times \mathcal U^N$. From this set, we can choose any solution: $\zeta^*(x, r) \in \bar\zeta^*(x, r)$. For any $(x, r) \in \Gamma$, the functions $\xi^*$ and $\mu^*$ are defined as the $\xi$- and $\mu$-components of $\zeta^*$:
\begin{equation} \label{eq:OCPsolution} \begin{multlined}
    \zeta^*(x, r) = (\xi^*(x, r), \mu^*(x, r))\\
    = (\xi_0^*, \xi_1^*, \dots, \xi_N^*, \mu_0^*, \mu_1^*, \dots, \mu_{N-1}^*).
\end{multlined}\end{equation}
The cost associated with $\zeta^*(x, r)$ is
\begin{equation}\label{eq:optimalCost}
    J(x, r) \defas V(\xi_N^*, r) +\! \sum_{i=0}^{N-1} \ell(\xi_i^*, \mu_i^*, r).
\end{equation}

With the OCP established, we have all the elements necessary to define the MPC feedback policy $\kappa\colon \Gamma \to \mathcal U$. Given any $(x, r) \in \Gamma$, the MPC control law $\kappa$ is
\begin{equation} \label{eq:MPCPolicy}
    \kappa(x, r) \defas \mu_0^*(x, r),
\end{equation}
where $\mu_0^*(x, r)$ selects the value $\mu_0^*$ from the OCP solution \ref{eq:OCPsolution}. To ensure that the OCP \ref{eq:OCP} guarantees an asymptotically stabilizing MPC feedback policy, we impose the following assumptions.

\begin{assume} \label{as:nonlinSysStability} \begin{assumelist}
    \item The stage cost $\ell$ is uniformly continuous and satisfies $\ell(\bar x_r,\allowbreak \bar u_r, r) = 0$, and there exists $\gamma \in K_\infty$ such that $\ell(x, u , r) \geq \gamma(\norm{x - \bar x_r})$ for all $(x, r) \in \Gamma$ and $u \in \mathcal U$. \label{as:nonlinStageCost} 

    \item\label{as:nolinTermCost} The terminal cost $V$ is uniformly continuous, and satisfies $V(\bar x_r, r) = 0$. Moreover, there exist a function $\iota \in \mathcal K$ and a terminal control law $\kappa_\subtext{T}\colon \termset \!\to \mathcal U$ such that for all $(x, r) \in \termset$,
    \begin{subequations}\begin{gather}
        V(x, r) \geq \iota(\norm{x - \bar x_r})\\
        V(x^\plus, r) - V(x, r) \leq -\ell(x, \kappa_\subtext{T}(x, r), r), \label{eq:termCostIneq}
    \end{gather}\end{subequations}
    where $x^\plus \defas f(x, \kappa_\subtext{T}(x, r))$.

    \item\label{as:termInvariant} $\termset$ is positively invariant and constraint admissible under the terminal dynamics, i.e.,
    \begin{equation*}\begin{multlined}
        (x, r) \in \termset \implies (x^\plus, r) \in \termset,\, x \in \mathcal X \text{ and } \kappa_\subtext{T}(x, r) \in \mathcal U.
    \end{multlined}\end{equation*}

    \item\label{as:equilInTermSet} $\termset$ contains all admissible equilibrium points in its interior, i.e., $(\bar x_r, r) \in \Int\termset$, for all $r \in \mathcal R_\epsilon$.

    \item\label{as:costContinuity} The optimal cost function $J$ is uniformly continuous in its domain $\Gamma$.
\end{assumelist}\end{assume}
Together, Assumptions~\ref{as:nonlinStageCost,as:nolinTermCost,as:termInvariant,as:equilInTermSet} ensure that the MPC controller guarantees safety and asymptotically stabilizes $\bar x_r$ \cite[Theorem~4.4.2]{ref:Goodwin2006}. Assumption~\ref{as:nonlinStageCost} ensures that the stage cost measures distance to the reference point while Assumptions~\ref{as:nolinTermCost,as:termInvariant,as:equilInTermSet} ensure that the system is locally controllable and admits a local Lyapunov function and region of attraction $\mathcal T$ for each $r \in \mathcal R_\epsilon$. Assumption~\ref{as:costContinuity} is essentially a robustness assumption; it is known to hold for LTI systems with convex constraints and for nonlinear systems with only input constraints. Unfortunately however, characterizing when Assumption~\ref{as:costContinuity} holds for general nonlinear systems is a research topic in its own right \cite{grimm2004examples}. 

\begin{rmk}[Terminal Ingredients]
    An MPC controller with terminal ingredients ($\termset$, $V$, $\kappa_\subtext{T}$) essentially enlarges the ROA of a local stabilizing controller $\kappa_\subtext{T}(\phdot, r)$. As such, finding these ingredients is equivalent to synthesizing a local controller for the nonlinear system $f$ which generally needs to be tackled on a case-by-case basis using knowledge about the system. More systematic techniques are available in some cases however. E.g., when $f$ is polynomial, the terminal ingredients can be computed simultaneously using sum-of-squares techniques \cite{jarvis2003some,cotorruelo2021reference}, if $f$ admits a controllable linearization about $(\bar x_r,\bar u_r)$ which does not depend on $r$, then $V$ and $\kappa_T$ can be determined using the linear-quadratic regulator and a suitable $\termset$ can be computed using sum-of-squares verification \cite{jarvis2003some}, or if $f$ admits controllable linearizations that depends on $r$ then techniques from linear parameter varying control are applicable \cite{mohammadpour2012control}.
\end{rmk}
 
Although MPC is a powerful control strategy, its main limitation is that the OCP~\ref{eq:OCP} is infeasible if the state $x$ cannot be steered to the terminal set $\termset$ in $N$ steps, i.e., $(x, r) \notin \Gamma$. For targets beyond the prediction horizon, a common strategy is to increase $N$. However, this approach is computationally expensive, potentially precluding real-time implementation, and is also vulnerable to getting ``stuck'' in local minima in non-convex environments.

\subsection{Path Feasibility Governor}

One way to avoid getting stuck in local minima and to reduce the computational burden associated with long horizons is to pair the controller with a path planner. However, integrating path planners with MPC in a rigorous way is challenging. Our solution to this problem is the PathFG, an add-on unit that dynamically manipulates the references passed to the MPC controller and ensures that these references are pairwise reachable within the MPC prediction horizon. The PathFG essentially filters the progress of the system along the path in a way that guarantees constraint satisfaction, asymptotic stability, and recursive feasibility of MPC.

The PathFG is compatible with any path planner that, when given $x_0 \in \mathbb R^{n_x}$ and $r \in \mathcal R_\epsilon$, produces a feasible path $\mathcal P(x_0, r)$ satisfying
\begin{equation}\begin{gathered}
    \mathcal P(x_0, r) = \set{\p(s) | s \in [0, 1]} \label{eq:path}\\[0.3ex]
    \begin{aligned}
        \text{s.t.} \quad& \p\colon [0, 1] \to \mathcal R_\epsilon \text{ is continuous}\\
            & (x_0, \p(0)) \in \Gamma,\, \p(1) = r.
    \end{aligned}
\end{gathered}\end{equation}
This is only possible if such a feasible path exists in the first place, i.e., if $(x_0, r)$ is in
\begin{equation} \label{eq:D}
    \mathcal D \defas \set{(x_0, r) | r \in \mathcal R_\epsilon,\, \exists \mathcal P(x_0, r) \text{ s.t. \ref{eq:path}}}.
\end{equation}

\begin{assume} \label{as:path}
    Given $(x_0, r) \in \mathcal D$, the path planner returns a feasible path $\mathcal P(x_0, r)$ that satisfies \ref{eq:path}.
\end{assume}

The continuous function $\p\colon [0, 1] \to \mathcal R_\epsilon$ maps an auxiliary reference $s \in [0, 1]$ to the set of strictly steady-state admissible references $\mathcal R_\epsilon$ defined in \ref{eq:admissibleRef}. Importantly, every point  $\p(s)$ on the planned path must satisfy $\p(s) \in \mathcal R_\epsilon$ for all $ s \in [0, 1]$. This guarantees that each corresponding equilibrium point is constraint-admissible, i.e., $(\bar x_{\p(s)}, \bar u_{\p(s)}) \in \mathcal X \times \mathcal U$ for all $ s \in [0, 1]$. The target reference $r$ coincides with the endpoint of the path $\p(1)$. In dynamic environments, the path is replanned whenever the obstacle positions or the target reference $r$ are modified.

With the auxiliary reference $s$ introduced, we now define the equilibrium state and control input functions $\bar x_{(\cdot)}$ and $\bar u_{(\cdot)}$ in terms of $s$ as
\begin{subequations}
\begin{equation}
    \tilde x_s \defas \bar x_{\p(s)} \conjtext{and} \tilde u_s \defas \bar u_{\p(s)}.
\end{equation}
We also introduce variants of the OCP-related functions and sets in terms of $(x, s)$ pairs as follows:
\begin{gather} 
    \tilde J(x, s) \defas J(x, \p(s)) \label{eq:tildeJ}\\
    \tilde\termset \defas \set{(x, s) | (x, \p(s)) \in \termset} \label{eq:tildeTermset}\\
    \tilde\Gamma \defas \set{(x, s) | (x, \p(s)) \in \Gamma}\\
    \tilde\xi_N^*(x, s) \defas \xi_N^*(x, \p(s)) \label{eq:optimalState}\\
    \tilde\kappa(x, s) \defas \kappa(x, \p(s)) \label{eq:MPCctrlLaw}\\
    \tilde\kappa_\subtext{T}(x, s) \defas \kappa_\subtext{T}(x, \p(s)). \label{eq:termCtrlLaw}
\end{gather}
\end{subequations}

We are now ready to introduce the PathFG policy. The idea behind the PathFG is to pick the next auxiliary reference $s_k$ along the path as far as possible such that $(x_k, s_k) \in \tilde\Gamma$. Then the MPC remains feasible and will take care of the rest. In practice, computing an explicit representation of $\tilde \Gamma$ is impossible for realistic problems. Instead, we use the predicted trajectory from the previous timestep to build an implicit under-approximation. 

\begin{figure*}
    \centering\large\color{textcolor}
    \makebox[\linewidth]{%
    \scalebox{0.9}{\begin{tikzpicture}
        \node[anchor=south west, inner sep=0] (image) at (0,0) {\includegraphics[width=0.49\linewidth, trim=0 177.66 0 0, clip]{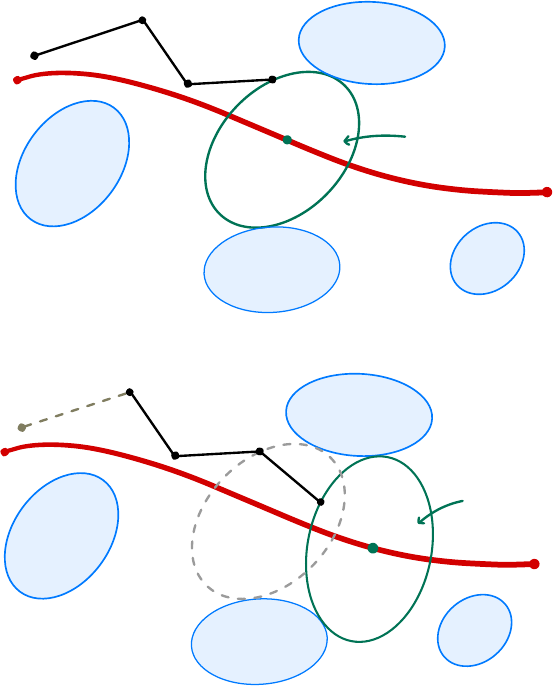}};
        \begin{scope}[x={(image.south east)}, y={(image.north west)}, inner sep=1pt, yscale=1/(1-0.54), shift={(0,-0.54)}]
            \node[anchor=north, goodnotesblue] at (0.675, 0.95) {$\mathcal O_1$};
            \node[anchor=north, goodnotesblue] at (0.129, 0.78) {$\mathcal O_2$};
            \node[anchor=south, goodnotesblue] at (0.5, 0.585) {$\mathcal O_3$};
            \node[anchor=south, goodnotesblue] at (0.88, 0.6) {$\mathcal O_4$};
            \node[anchor=north, red!70!textcolor] at (0.02, 0.87) {$\p(0)$};
            \node[anchor=south] at (0.03, 0.92) {$\xi_{0|k-1}^*$};
            \node[anchor=south] at (0.26, 0.97) {$\xi_{1|k-1}^*$};
            \node[anchor=north west] at (0.5, 0.88) {$\xi_{N|k-1}^*$};
            \node[anchor=north, green!40!textcolor] at (0.52, 0.79) {$\p(s_k)$};
            \node[anchor=west, green!40!textcolor] at (0.74, 0.8) {$\tilde{\mathcal X}\_f(s_k)$};
            \node[anchor=south, red!70!textcolor] at (0.96, 0.73) {$\p(1) = r$};
        \end{scope}
    \end{tikzpicture}}\
    \scalebox{0.9}{\begin{tikzpicture}
        \node[anchor=south west, inner sep=0] (image) at (0,0) {\includegraphics[width=0.49\linewidth, trim=0 0 0 164.5, clip]{figures/pathFG_3.pdf}};
        \begin{scope}[x={(image.south east)}, y={(image.north west)}, inner sep=1pt, yscale=2]
            \node[anchor=north, goodnotesblue] at (0.65, 0.412) {$\mathcal O_1$};
            \node[anchor=north, goodnotesblue] at (0.122, 0.235) {$\mathcal O_2$};
            \node[anchor=south, goodnotesblue] at (0.48, 0.04) {$\mathcal O_3$};
            \node[anchor=south, goodnotesblue] at (0.86, 0.06) {$\mathcal O_4$};
            \node[anchor=south west, green!40!textcolor] at (0.84, 0.25) {$\tilde{\mathcal X}\_f(s_{k+1})$};
            \node[anchor=south west] at (0.24, 0.42) {$\xi_{0|k}^*$};
            \node[anchor=north] at (0.31, 0.33) {$\xi_{1|k}^*$};
            \node[anchor=west] at (0.59, 0.26) {$\xi_{N|k}^*$};
            \node[anchor=north, green!40!textcolor] at (0.665, 0.2) {$\p(s_{k+1})$};
            \node[anchor=south, red!70!textcolor] at (0.96, 0.185) {$\p(1) = r$};
            \node[anchor=north, red!70!textcolor] at (0.01, 0.33) {$\p(0)$};
        \end{scope}
    \end{tikzpicture}}%
    }%
    \medskip
    \caption{\lineskiplimit=-10pt\linespread{1.1}\selectfont Illustration of the PathFG workflow, showing the system at timestep $k$ (left) and at timestep $k+1$ (right). Obstacles $\mathcal O_1$ through $\mathcal O_4$ and the path from the start $\p(0)$ to the goal $\p(1) = r$ are also indicated. At timestep $k$, the predicted states from the last timestep, $\xi_{0|k-1}, \xi_{1|k-1}^*, \dots, \xi_{N|k-1}^*$, have already been used to calculate the current auxiliary reference $s_k$. A slice of the terminal set at $s_k$ is shown: $\tilde{\mathcal X}\_f(s_k) \defas \set{x | (x, s_k) \in \tilde\termset}$. The auxiliary reference is computed via $s_k = g(\xi_{N|k-1}^*)$, which selects the furthest auxiliary reference on the path such that the optimal final predicted state $\xi_{N|k-1}^*$ lies inside $\tilde{\mathcal X}\_f(s_k)$. The system advances to timestep $k + 1$, where the predicted trajectory from timestep $k$ is marked as: $\xi_{0|k}^* = \xi_{1|k-1}^*$, $\xi_{1|k}^* = \xi_{2|k-1}^*$ and so on. The $(k + 1)\+$th auxiliary reference $s_{k+1}$ is computed analogously through $s_{k+1} = g(\xi_{N|k}^*)$. This again selects the furthest auxiliary reference $s_{k+1}$ satisfying $\xi_{N|k}^* \in \tilde{\mathcal X}\_f(s_{k+1})$.}
    \label{fig:pathFG}
\end{figure*}

The process is illustrated in Figure~\ref{fig:pathFG}. We first extract $\xi_{N|k-1}^* \defas \tilde \xi_N^*(x_{k-1}, s_{k-1})$, the final predicted state from the previous timestep, as defined in \ref{eq:optimalState}. We then pick the new auxiliary reference $s_k$ as the largest value of $s$ such that $\xi_{N|k-1}^*$ remains in the terminal set around $\tilde x_s$, formally
\begin{equation} \label{eq:PathFGDef}
    s_k = g(\xi_{N|k-1}^*) \defas \max\set{s \in [0, 1] | (\xi_{N|k-1}^*, s) \in \tilde\termset}.
\end{equation}
Since $\xi_{N|k-1}^*$ was reachable in $N$ steps from $x_{k-1}$, it is reachable in $N - 1$ steps from $x_k$, which guarantees feasibility. The auxiliary reference is then passed to the MPC controller to generate a control action
\begin{equation} \label{eq:tildeKappa}
    u_k = \tilde\kappa(x_k, s_k) 
\end{equation}
and the process continues. 

Putting it all together, the PathFG update equation is
\begin{equation}\begin{multlined}
    s_k = \tilde g(x_{k-1}, s_{k-1}) \defas g(\tilde \xi_{N}^*(x_{k-1},s_{k-1}))
\end{multlined}\end{equation}

In Section~\ref{sec:theory}, we formally prove that this choice of auxiliary reference update ensures recursive feasibility of the MPC and that the reference $\tilde x_s$ converges to $\bar x_r$ in finite time.

\begin{rmk}
    FGs built using the proposed underapproximation of $\tilde\Gamma$ performs well when compared with FGs built using $\tilde\Gamma$ itself as illustrated in \cite[Fig. 6]{ref:Skibik2023}.
\end{rmk}

The PathFG framework offers several key advantages. 1)~It enables the delegation of non-convex constraints to a path planner. This reduces computational complexity and avoids the issue of the controller getting stuck in local minima. 2)~Its modular design offers flexibility in the choice of path planners and MPC formulations. 3)~It significantly enlarges the ROA of MPC, increasing it from the $N$-step backwards reachable set of $\termset$, to
\begin{equation} \label{eq:ROA}
    \mathcal D_x(r) \defas \set{x_0 | (x_0, r) \in \mathcal D},
\end{equation}
the set of all initial states from which a feasible path $\mathcal P(x_0, r)$ to the target $r$ can be generated which significantly reduces the need for long prediction horizons 4)~The PathFG incurs minimal computational overhead due to its one-dimensional parameter search space when solving \ref{eq:PathFGDef}, enabling efficient computation using simple algorithms like the bisection method. Taken together, PathFG effectively addresses the integration challenge between the MPC and path planners in a scalable manner. In the following section, we present a formal theoretical analysis demonstrating that the framework satisfies all control objectives.

\begin{rmk}[Robustness]
    As an add-on unit, the PathFG tends to inherit the robustness properties of the underlying MPC controller and can be applied in the presence of bounded disturbances by 1)~replacing \ref{eq:MPCPolicy} with a controller that is robustly recursively feasible and 2)~ensuring that $\termset$ is robustly positively invariant. Methods for designing MPC controllers and terminal sets with these properties are available for both linear and nonlinear systems (see e.g., \cite{mayne2005robust,kouvaritakis2016model,mayne2011tube,cannon2011robust,yu2013tube}). The PathFG can also be used with a nominal MPC controller such as \ref{eq:MPCPolicy} in the presence of disturbances if it is input-to-state-stable \cite{limon2009input} (guaranteed by Assumption~\ref{as:nonlinSysStability}) and the designer softens the constraints to ensure recursive feasibility \cite{zeilinger2014soft}.
\end{rmk}

\section{Theoretical Analysis} \label{sec:theory}

In this section, we examine the properties of the closed-loop system \ref{eq:nonlinearsysEqs} under the combined PathFG+MPC feedback policy. For notational simplicity, we omit the constant target reference $r$ in the proof. The results readily extend to piecewise constant references.

Given a path $\mathcal P(x_0, r)$, the closed-loop dynamics of \ref{eq:nonlinearsysEqs} under the combined PathFG+MPC policy are
\begin{subequations} \label{eq:compactDynamics} \begin{align}
    s_{k+1} &= \tilde g(x_k, s_k) \defas g(\tilde\xi_N^*(x_k, s_k)) \label{eq:compactPathFGDynamics}\\
    x_{k+1} &= \tilde f(x_k, s_k) \defas f(x_k, \tilde\kappa(x_k, s_k)), \label{eq:noPathFGDynamics}
\end{align}\end{subequations}
where $\tilde\xi_N^*$ is the optimal final predicted state \ref{eq:optimalState}, and $\tilde\kappa$ is the MPC control law \ref{eq:MPCctrlLaw}, with initial condition $x_0$ and $s_0 = g(x_0)$.

We now show that the combined PathFG+MPC feedback policy satisfies the control objectives of safety and asymptotic stability stated in Section~\ref{sec:problemFormulation}.

The following theorem provides sufficient conditions under which the combined PathFG+MPC control law ensures recursive feasibility and satisfies the \emph{safety} control objective.

\begin{theorem}[Safety and Feasibility] \label{prop:safety}
    Let Assumptions~\ref{as:system,as:nonlinSysStability,as:path} hold. Then for any $r \in \mathcal R_\epsilon$ and $x_0 \in \mathcal D_x(r)$, the solution trajectory $\qty{(x_k, s_k)}_{k=0}^\infty$ to \ref{eq:compactDynamics} is well defined, and satisfies $(x_k, s_k) \in \tilde\Gamma$, $x_k \in \mathcal X$, and $u_k \in \mathcal U$ for all $k \in \mathbb N$.
\end{theorem}

\begin{proof}
    See Appendix~\ref{sec:proofOfsafety}.
\end{proof}

Theorem~\ref{prop:safety} implies that, for any initial position and any admissible target $r$ that can be connected by a path, the closed-loop system governed by the combined PathFG+MPC control law is guaranteed to satisfy state and control input constraints, and the OCP is recursively feasible for all time.

Next, we prove that the combined PathFG+MPC feedback policy satisfies the \emph{asymptotic stability} control objective.

\begin{theorem}[Path Tracking] \label{prop:AS}
    Given Assumptions~\ref{as:system,as:nonlinSysStability,as:path}, for any $r \in \mathcal R_\epsilon$ and $x_0 \in \mathcal D_x(r)$ there exists $\beta \in \mathcal{KL}$, $\gamma \in \mathcal K$ and finite $k^* \in \mathbb N$ such that the closed-loop system \ref{eq:compactDynamics} satisfies $s_k = 1$  $\forall k \geq k^*$ and
    \begin{equation}
        \norm{x_k - \tilde x_{s_k}} \leq \beta(\norm{x_0 - \tilde x_{s_0}}, k) + \gamma\qty\Big(\sup_{j\leq k} \norm{s_{j+1} - s_j})
    \end{equation}
    for all $k \in \mathbb N$.
\end{theorem}
 
\begin{proof}
    See Appendix~\ref{sec:proofOfAS}.
\end{proof}

Theorem~\ref{prop:AS} proves that the combined PathFG+MPC controller causes the state to track the path $\tilde x_s$ with a bounded tracking error and produces the largest possible ROA given the constraints which implies that the PathFG+MPC controller works for any $(x_0, r) \in \mathcal D$ in \ref{eq:D}. A classical form of asymptotic stability follows immediately. Note that $\beta$, $\gamma$ and $k^*$ in Theorem~\ref{prop:AS} depends on the chosen path $\mathcal P(x_0, r)$; paths passing closer to obstacles generally yield larger $k^*$.

\begin{corollary}[Asymptotic Stability]
    Under the same assumptions as Theorem~\ref{prop:AS}, we have $\norm{x_k - \bar x_r} \leq \beta(\norm{x_{k^*} - \bar x_r}, k - k^*)$ for all $k \geq k^*$.
\end{corollary}

Theorems~\ref{prop:safety,prop:AS} prove that the combined PathFG+MPC feedback policy satisfies both the \emph{safety} and \emph{asymptotic stability} control objectives. Therefore, for any initial state $x_0$, from which a feasible path $\mathcal P(x_0, r)$ to the target $r$ exists, the PathFG generates a sequence of auxiliary references $s$ that gradually guides the system along the path to the target without violating any constraint. The combined control law significantly expands the ROA from $\set{x | (x, r) \in \Gamma}$, the set from which the terminal set can be reached within $N$ timesteps, to $\mathcal D_x(r)$, which includes all initial states that admit a feasible path to the target $r$. The expansion of the ROA is accomplished by leveraging a path planner to generate feasible paths, thus overcoming the challenges imposed by non-convex constraints.

\section{Motion Planning for a Unicycle} \label{sec:Unicycle}

In this section we apply the PathFG to a motion planning problem for the discrete-time unicycle model
\begin{equation}
    x_{k+1} = x_k + v_k \mat[\cos(\theta_k)\\\sin(\theta_k)]
\end{equation}
with states $x \in \mathbb R^2$ and heading and speed inputs $u = (\theta, v)$ as illustrated in Figure~\ref{fig:unicycle}.

\begin{figure}
    \centering
    \includegraphics[width=\linewidth]{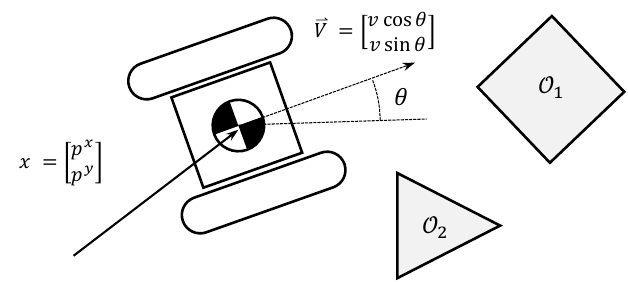}
    \caption{A unicycle navigates a cluttered environment.}
    \label{fig:unicycle}
\end{figure}

\subsection{Control Problem}

Our goal is to reach a position $r \in \mathbb R^2$ subject to the obstacle avoidance constraint
\begin{equation}
      x\in \mathcal X = \bigcap_{j=1}^m \mathcal O_j^\complement,
\end{equation}
where $\mathcal O_j = \set{x | C_j\,x \leq d_j} \subset \mathbb R^2$ are convex polyhedral obstacles, and the input constraints $u \in \mathcal U = (-\pi, \pi] \times [0, \bar v]$. We can select the equilibrium manifolds $\bar x_r = r = (r_x, r_y)$, $\bar u_r = 0$, stage cost $\ell(x, u) = q\,\norm{x - r}^2$, terminal cost $V(x, r) = \norm{x - r}^2$, terminal controller 
\begin{equation}
    \kappa_\subtext{T}(x, r) = \mat[\angle(r-x)\\k\,\norm{x - r}]
\end{equation}
where $\angle([x\; y]^\T) \defas \operatorname{atan2}(y, x)$ is the angle operator and $k \in (0, 1)$ is a control gain. This controller points the unicycle directly towards $r$ and moves towards $r$ at a speed proportional to the remaining distance.

It is straightforward to show that for all $x$ and $r$ such that $k\,\norm{x - r} \leq \bar v$,
\begin{equation}
    V\qty\big(x^\plus, r) - V(x, r) = -k\,(2 - k)\,\norm{x - r}^2
\end{equation}
and hence Assumption~\ref{as:nonlinSysStability} holds as long as $k \in (0, 1)$, and $q \leq k\,(2 - k)$. Following \cite{nicotra2018explicit} we construct the terminal set by computing the largest level set of the Lyapunov function such that the constraints are satisfied which in this scenario is given by
\begin{equation}
    \termset = \set{(x, r) | V(x,r) \leq \min\qty*{(\bar v/k)^2, d(r)^2}},
\end{equation}
where $d(r)$ denotes the Euclidean distance from the reference to the nearest obstacle, which can be computed using, e.g., the Gilbert–Johnson–Keerthi (GJK) algorithm for convex polyhedra \cite{Gilbert1988}.

\begin{figure}
    \includegraphics[width=\linewidth, center]{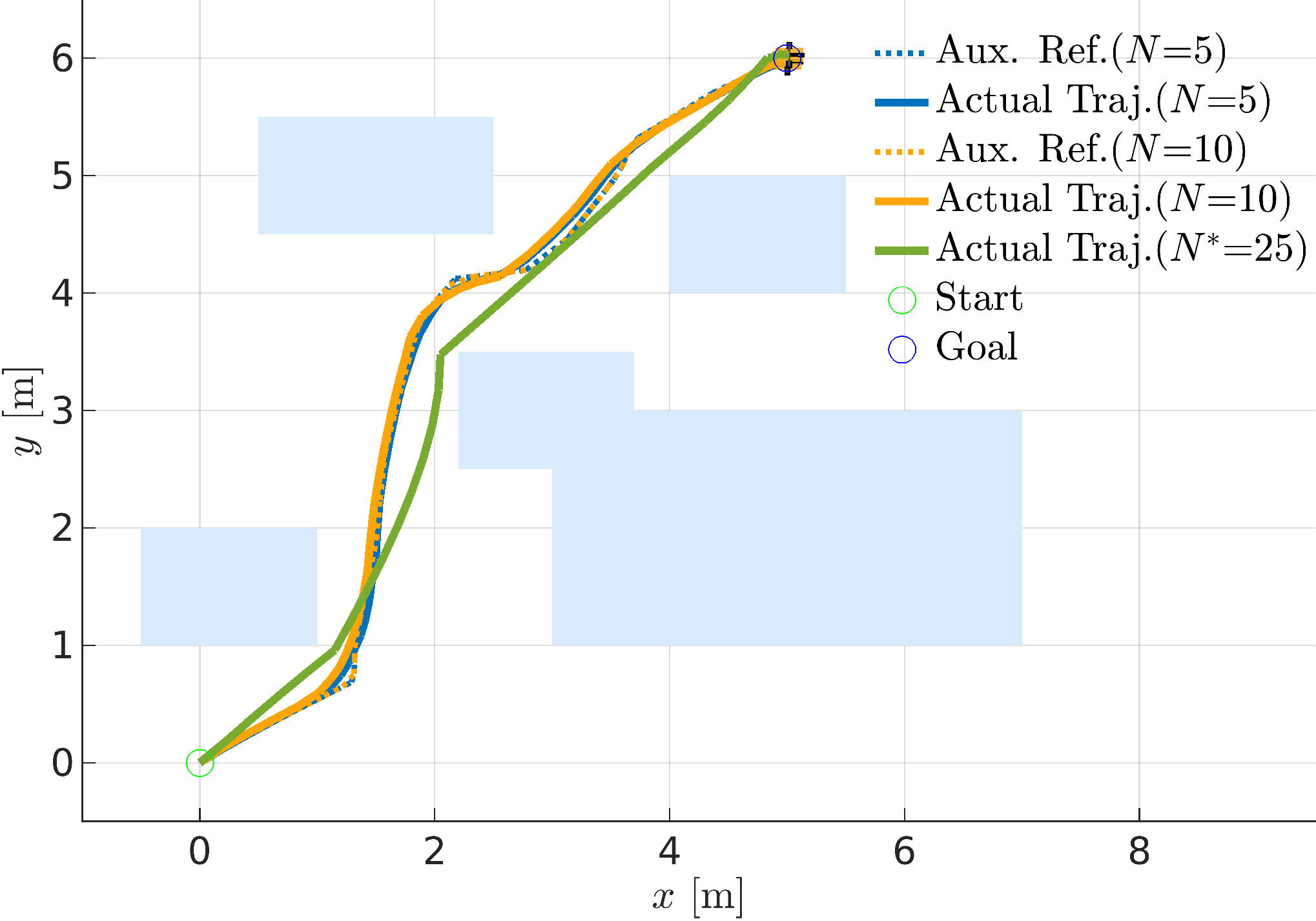}
    \medskip
    \caption{Unicycle navigation with obstacle avoidance using different control strategies. The PathFG+MPC controller is evaluated for horizons $N = 5$ and $N = 10$, while ungoverned MPC (i.e., without PathFG) is tested at $N^* = 25$.} \label{fig:unicylceNavigation}
\end{figure}

\begin{figure}
    \includegraphics[width=1\linewidth, center]{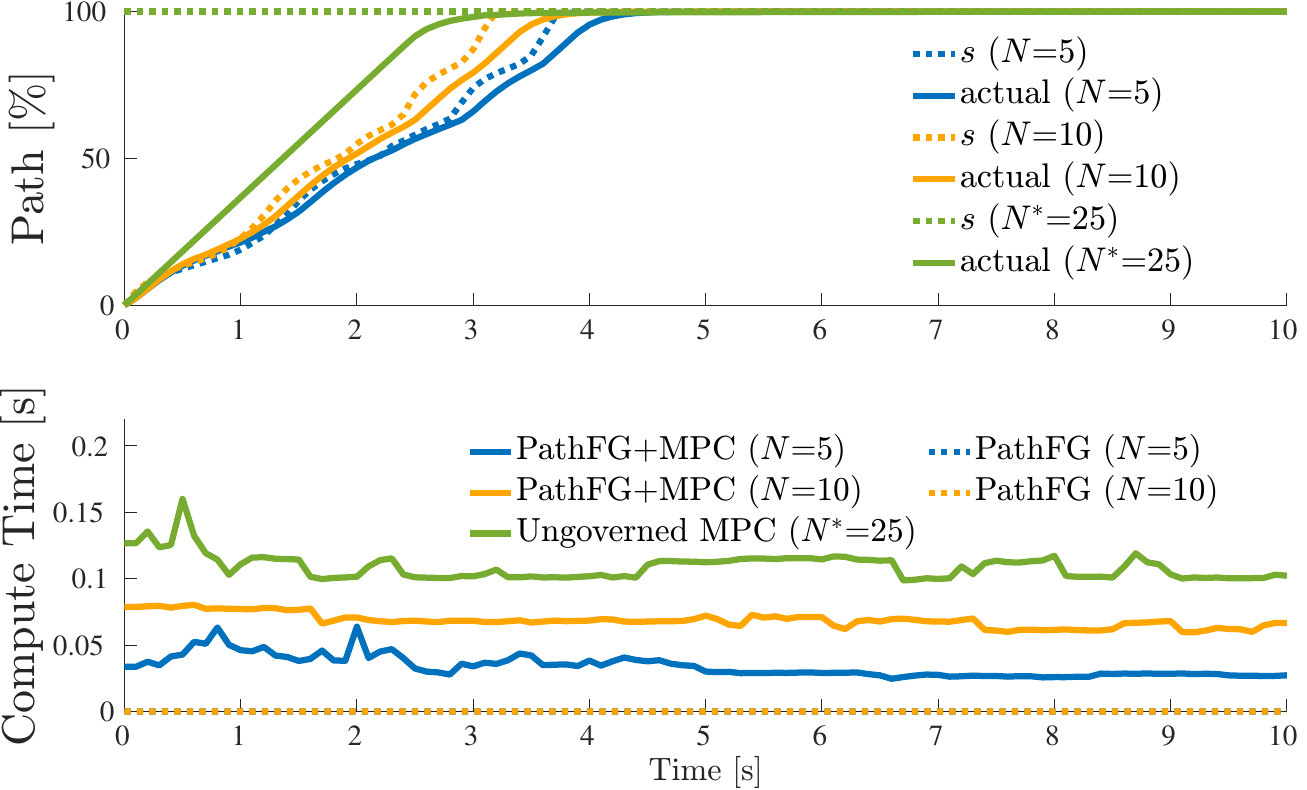}\medskip
    \caption{Comparison of closed-loop performance between PathFG+MPC and ungoverned MPC. Subfigure~1 shows that the actual trajectories reach the target as the guiding auxiliary reference $s$ converges to $100\%$. The ungoverned MPC ($N^* = 25$) achieves convergence first, followed by the short-horizon PathFG+MPC ($N = 5$ and $N = 10$). Subfigure~2 illustrates that ungoverned MPC incurs significantly higher computational cost than PathFG+MPC.} \label{fig:unicylcePerformance}
\end{figure}

\subsection{Simulation Results}

In this section, we evaluate the performance of the PathFG framework using MATLAB/Simulink simulations of unicycle navigation in an obstacle-cluttered environment. All simulations were executed on an ASUS G614J laptop (2.2 GHz Intel i9, 32 GB RAM) running MATLAB R2024b.

Figure~\ref{fig:unicylceNavigation} shows that PathFG+MPC with short horizons ($N = 5$ and $N = 10$) achieves collision-free navigation in this obstacle-dense environment, whereas ungoverned MPC (i.e., without PathFG) requires much longer horizons ($N^* = 25$), where $N^*$ denotes the minimum horizon needed to ensure feasibility of the OCP under the given constraints and target reference. Figure~\ref{fig:unicylcePerformance} compares their closed-loop performance, with the results highlighting the performance-computation trade-off: long-horizon MPC produces more aggressive trajectories at the cost of increased computation time, while the PathFG component introduces negligible computational overhead (${\sim}10^{-5}\,\mathrm{s}$ per update).

\section{Motion Planning for\texorpdfstring{\\}{ }LTI Systems with Obstacles} \label{sec:implementation}

In constrained and/or safety-critical environments, solving non-convex optimization problems online can be difficult or even impossible, particularly when obstacle constraints are involved. A common ad-hoc approach to alleviate this burden is to linearize obstacle constraints around a reference path. In this section, we illustrate how PathFG provides a principled and theoretically sound justification for this strategy by separating obstacle avoidance from control, yielding a fully convex MPC optimization problem that can be solved efficiently and reliably. A summary of the procedure is given in Algorithm \ref{fig:alg}.

\begin{algorithm}[tb]
    \caption{Combined PathFG+MPC law}
    \label{fig:alg}
    Inputs: $\xi_{(\phdot)|k-1}^*$, $x_k$
    \begin{algorithmic}[1]
        \Substeps{$s_k = g(\xi_{N|k-1}^*)$ as defined in \ref{eq:PathFGDef}}
            \State $s = 1$
            \While{$\Delta > 0$}
                \State select a candidate $s$ by optimizer
                \State compute linearized constraints at $\tilde x_s$ as in \ref{eq:linearConstraint}
                \State compute $\Delta(\smash[b]{\xi_{N|k-1}^*}, \p(s))$ as in \ref{eq:computeDelta}
            \EndWhile
        \EndSubsteps
        \Substeps{\textit{{Optional}}: linearize obstacle constraints}
            \For{each $\xi_{(\cdot)|k-1}^*$}
                \State project reference point onto obstacle as in~\ref{eq:projection} \State compute half-space as in~\ref{eq:halfspace}
            \EndFor
        \EndSubsteps    
        \State $u_k = \tilde \kappa(x_k, s_k)$ where $\tilde \kappa$ is defined in \ref{eq:tildeKappa}
    \end{algorithmic}
\end{algorithm}

\subsection{Control Problem Formulation} \label{sec:linProblemFormulation}

We specialize to the case where the system in \ref{eq:nonlinearsysEqs} is linear and time-invariant:
\begin{equation} \label{eq:linearSysEq}
    x_{k+1} = f(x_k, u_k) = Ax_k + Bu_k,
\end{equation}    
where $A \in \mathbb R^{n_x\times n_x}\!$ and $B \in \mathbb R^{n_x\times n_u}\!$ are the system matrices. 

\begin{assume} \label{as:stabilizable}
    The pair $(A, B)$ is stabilizable.
\end{assume}

Assumption~\ref{as:stabilizable} is a necessary and sufficient condition \cite{ref:Limon2008} for the existence of a non-trivial kernel of
\begin{equation} \label{eq:Z}
    Z \equiv [A - I\,\ B].
\end{equation}
This implies that the system \ref{eq:linearSysEq} admits a family of equilibrium points satisfying
\begin{equation} \label{eq:linearEquilibrium}
    Z \mat[\bar x_r\\\bar u_r] = 0.
\end{equation}
This property allows us to write any solution of \ref{eq:linearEquilibrium} as
\begin{equation}
    \bar x_r = G_x r, \quad \bar u_r = G_u r,
\end{equation}
where the columns of $G \equiv [G_x^\T*[0]\; G_u^\T[0]]^\T$ form a basis for $\operatorname{ker} Z$. 

The system \ref{eq:linearSysEq} is subject to pointwise-in-time constraints:
\begin{equation} \label{eq:stateAndControlConstraints}
    x_k \in \mathcal X = \mathcal C \cap \qty\bigg(\,\bigcap_{j=1}^m \mathcal O_j^\complement), \quad u_k \in \mathcal U,
\end{equation}
where $\mathcal C \subseteq \mathcal X$ and $\mathcal U$ are compact polyhedra, $\mathcal O_j \subseteq \mathbb R^{n_x}$ represent convex obstacles inflated by a safety margin $\oldepsilon$.

With the system and constraints formulated, we next specialize the general nonlinear OCP in \ref{eq:OCP} to a linear formulation
\begin{equation}\begin{gathered}
    \begin{multlined} \label{eq:linearOCP}
        \zeta^*(x_k, r) = \argmin_{\xi,\mu} \biggl[\norm{\xi_N - \bar x_r}_P^2\\[-1.5ex]
        +\! \sum_{i=0}^{N-1} (\norm{\xi_i - \bar x_r}_Q^2 + \norm{\mu_i - \bar u_r}_R^2)\biggr]
    \end{multlined}\\
    \begin{aligned}
        \text{s.t.} \quad & \xi_0 = x_k\\
            & \xi_{i+1} = A\xi_i + B\mu_i, && \forall i \in \mathbb N_{[0,N-1]}\\
            & \xi_i \in \mathcal X,\, \mu_i \in \mathcal U, && \forall i \in \mathbb N_{[0,N-1]}\\
            & (\xi_N, r) \in \termset,
    \end{aligned}
\end{gathered}\end{equation}
where $Q \in \mathbb R^{n_x \times n_x}\!$ and $R \in \mathbb R^{n_u \times n_u}$\! are the stage cost matrices, $P \in \mathbb R^{n_x \times n_x}\!$ is the terminal cost matrix. The terminal control law is chosen as
\begin{equation} \label{eq:lintermlaw}
    \kappa_\subtext{T}(x, r) = \bar u_r - K\,(x - \bar x_r),
\end{equation}
where $K \in \mathbb R^{n_u \times n_x}\!$ is a gain matrix. The following assumption ensures the linear OCP generates a stabilizing control law.

\begin{assume} \label{as:linSysStability} \begin{assumelist}
    \item\label{as:linStageCost} The stage cost matrices satisfy $Q = Q^\T \succ 0$ and $R = R^\T \succ 0$.
    \item\label{as:decreasingLyapunov} The cost and gain matrices satisfy
    \begin{equation} \label{eq:LMI}
        -(K^\T RK + Q) \succeq (A - BK)^\T* P\+(A - BK) - P
    \end{equation}
    and $P = P^\T[2] \succ 0$.
\end{assumelist}\end{assume}

\begin{lemma} \label{prop:linProperties}   
   Let Assumption~\ref{as:linSysStability} hold, then Assumptions~\ref{as:nonlinStageCost,as:nolinTermCost,as:costContinuity} are satisfied.
\end{lemma}

\begin{proof}
    See Appendix~\ref{sec:proofOflinProperties}.
\end{proof}

\begin{figure}
    \centering
    \includegraphics[width=0.9\linewidth]{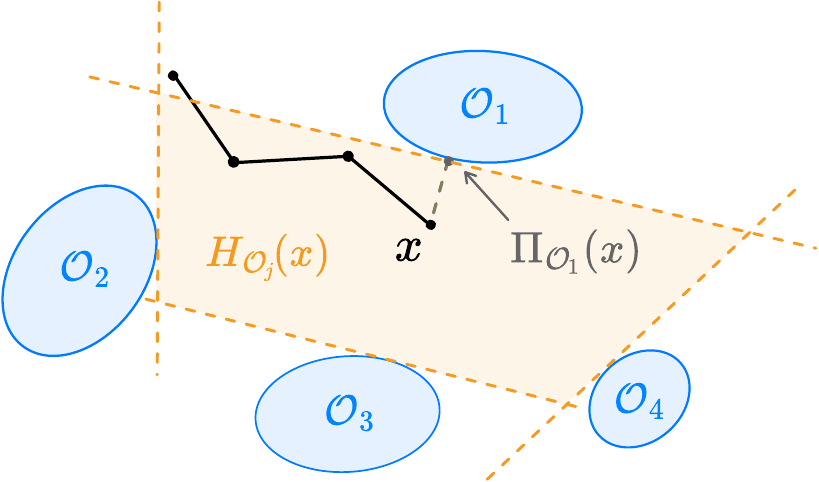}
    \medskip
    \caption{An illustration of the polyhedral freespace obtained by ``linearizing'' each obstacle $\mathcal O_j$ about the point $x$ by replacing the obstacle with a hyperplane supporting the point $\Pi_{\mathcal O_j}(x)$.}
    \label{fig:halfspace}
\end{figure}

The obstacle avoidance constraints $\mathcal O_j^\complement$ are non-convex, which can be readily handled by the path planner. However, incorporating non-convex obstacle avoidance constraints in the MPC can make solving the OCP more expensive and less reliable. To address this, we use affine over-approximations of the obstacles so that the MPC handles only convex constraints while the path planner manages the non-convexities. The linearization process is illustrated in Figure~\ref{fig:halfspace}. The projection of a point $x$ onto an obstacle is denoted by
\begin{equation} \label{eq:projection}
    \Pi_{\mathcal O_j}(x) \defas \argmin_{y \in \mathcal O_j} \norm*{x - y}.
\end{equation}
We can then construct a half-space over-approximation of the obstacle as
\begin{equation} \label{eq:halfspace}
    H_{\mathcal O_j}(x) \defas \set{y | [x - \Pi_{\mathcal O_j}(x)]^\T\, [y - \Pi_{\mathcal O_j}(x)] \geq 0}.
\end{equation}

\begin{lemma} \label{prop:linObsProperties}
    The linearized obstacle constraints $H_{\mathcal O_j}(x)$ in \ref{eq:halfspace} satisfy:
    \begin{enumerate}[label=\Alph{*}), ref=\arabic{lemma}\Alph{*}]
        \item\label{prop:linObsSafe} For any $x \in \Int\mathcal O_j^\complement$, $y \in H_{\mathcal O_j}(x)$ implies $y \in \mathcal O_j^\complement$.

        \item\label{prop:linObsInterior} $x \in \Int\mathcal O_j^\complement$ implies $x \in \Int[H_{\mathcal O_j}(x)]$.
    \end{enumerate}
\end{lemma}

\begin{proof}
    See Appendix~\ref{sec:proofOflinObsProperties}.
\end{proof}

To move the linearization process outside the OCP, we approximate the obstacle constraints around the previous predicted trajectory $\xi_{(\cdot)|k-1}^*$. By doing so, the resulting linearized state constraints collectively form a polyhedral free space for the predicted state $\xi_{i|k}$ as
\begin{equation} \label{eq:polyhedralConstraints}
    \xi_{i|k} \in \mathcal X_{i|k} \defas \mathcal C \cap \qty\bigg[\,\bigcap_{j=1}^m H_{\mathcal O_j}(\xi^*_{i+1|k-1})]
\end{equation}
for prediction index $i \in \mathbb N_{[0,N-1]}$ and timestep $k \in \mathbb N$. 

\begin{rmk}
    The linearization of obstacle constraints serves to illustrate how PathFG enables the MPC optimization problem to be fully “convexified” by shifting all the non-convexities to the planning level. The primary advantage of this approach is that convex optimization problems can be solved more reliably and efficiently in real time. This benefit, however, comes at the cost of reduced performance due to the reduction in available free space. The choice ultimately depends on the operating environment and the user's priorities. 
\end{rmk}

Any terminal set that satisfies Assumptions~\ref{as:termInvariant,as:equilInTermSet} is admissible. We again follow the Lyapunov-based approach in \cite{nicotra2018explicit} and compute a threshold value that renders the sublevel sets of
\begin{equation} \label{eq:LyapuValue}
     V(x, r) \defas \norm{x - \bar x_r}_P^2
\end{equation}
invariant and constraint admissible. Specifically, we pick
\begin{equation} \label{eq:termSetConstruct}
    \termset = \set{(x, r) | \Delta(x, r) \leq 0},
\end{equation}
where
\begin{equation} \label{eq:computeDelta}
    \Delta(x, r) \defas \max_{i\in\mathbb N_{[1,n]}}\! \,[V(x, r) - \Lambda_i(r)],
\end{equation}
and
\begin{equation} \label{eq:lyapunovThreshold}
    \Lambda_i(r) \defas \frac{[d_i(r) - c_i(r)^\T \bar x_r]^2}{c_i(r)^\T* P^\inv* c_i(r)}
\end{equation}
is the Lyapunov threshold value for a linear constraint $c_i(r)^\T x \leq d_i(r)$. Since the constraints in \ref{eq:stateAndControlConstraints} are either already linear or can be over-approximated by a half-space representation using \ref{eq:halfspace}, we can express the linearized constraints as
\begin{equation} \label{eq:linearConstraint} \begin{multlined}
    \set{x | c_i(r)^\T x \leq d_i(r),\, \forall i \in \mathbb N_{[1,n]}} =\\[-1.2ex]
    \hspace{3em} \mathcal C \cap \set{x | \kappa_\subtext{T}(x, r) \in \mathcal U} \cap \qty\bigg[\,\bigcap_{j=1}^m H_{\mathcal O_j}(\bar x_r)].
\end{multlined}\end{equation}

\begin{lemma} \label{prop:linTermSet}
    The terminal set $\termset$ constructed in \ref{eq:termSetConstruct} satisfies Assumptions~\ref{as:termInvariant,as:equilInTermSet}.
\end{lemma}

\begin{proof}
    See Appendix~\ref{sec:proofOflinTermSet}.
\end{proof}

\begin{rmk}
    Lemmas~\ref{prop:linProperties,prop:linObsProperties,prop:linTermSet} establish that the linear formulations satisfy Assumption~\ref{as:nonlinSysStability}. Since Assumptions~\ref{as:system} and~\ref{as:path} also hold in this setting, it follows that Theorems~\ref{prop:safety} and~\ref{prop:AS} remain valid.
\end{rmk}

\subsection{Quadrotor Navigation Example} \label{sec:numericalEx}
 
In this section, we numerically validate the PathFG framework through simulations of quadrotor navigation in obstacle-rich environments. The results demonstrate that the PathFG framework guarantees safety, achieves finite-time convergence to the target reference, and significantly reduces the computational burden of the MPC at the expense of a moderate performance reduction compared to the ungoverned MPC. 

\subsubsection{Quadrotor Model}

We consider a quadrotor navigating an obstacle-cluttered environment. For illustrative purposes, we adopt physical parameters based on the Crazyflie platform. The dynamic model follows \cite{ref:Convens2021}. Many quadrotors use a cascaded control architecture, as illustrated in Figure~\ref{fig:cascadedContArch}, with an inner loop that converts angular rate and total thrust signals into motor commands, and an outer loop that controls position, velocity, and orientation. Here, we focus on the outer loop; flight controllers for the inner loop are well-established.

\begin{figure}
    \centering
    \includegraphics[width=\linewidth]{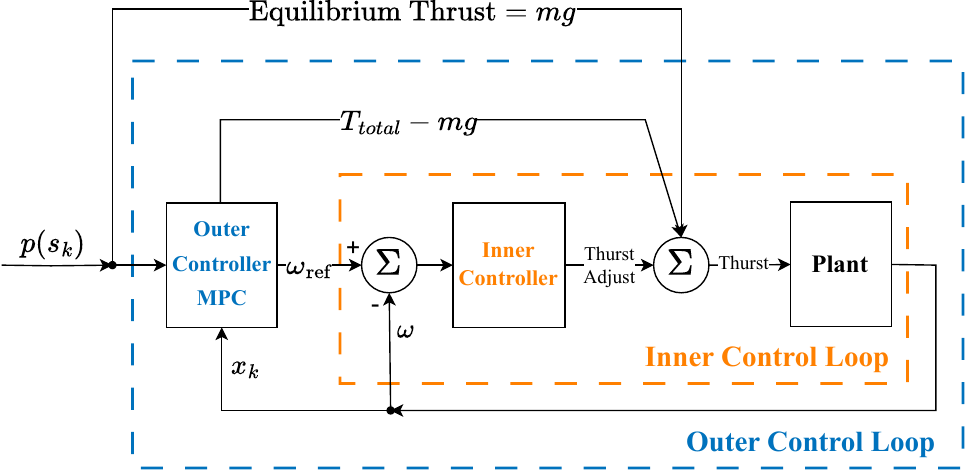}
    \smallskip
    \caption{Cascaded control architecture of the numerical example. The inner controller is a standard flight controller, while the outer controller is MPC governed by the PathFG.}
    \label{fig:cascadedContArch}
\end{figure}

We use the following state-space model for the outer-loop dynamics. The state vector is
\begin{equation}
    x = (p, \dot p, \psi) \in \mathbb R^9,
\end{equation}
in which $p \defas (p_x, p_y, p_z) \in \mathbb R^3$ and $\dot p \defas (\dot p_x, \dot p_y, \dot p_z) \in \mathbb R^3$ denote the position and the velocity of the body reference frame $\mathcal B$ with respect to the inertial frame $\mathcal I$, respectively, and $\psi \defas (\gamma, \beta, \alpha) \in \mathbb R^3$ denotes the Euler attitudes (roll, pitch and yaw, respectively) that re-align the axes of $\mathcal B$ with the axes of $\mathcal I$.

The control input is
\begin{equation}
    u = (T_\subtext{total}, \omega) \in \mathbb R^4,
\end{equation}
where $T_\subtext{total}$ is the total thrust generated by the four rotors, and $\omega \defas (\omega_x, \omega_y, \omega_z)$ are the angular velocities.

A natural parameterization of the equilibrium manifold is 
\begin{equation}
    \bar x_r = (r, 0, 0) = \Xi^\T r, \quad \bar u_r = (mg, 0)
\end{equation}
where $r$ is the reference for position $p$,  $\Xi \defas [I_3\ 0_{3\times6}] \in \mathbb R^{3\times9}$ denotes the projection matrix that extracts the agent's position from its full state. Here, $m = \SI{32}{g}$ is the mass of the agent, and $g \approx \SI{9.81}{m/s^2}$ is the gravitational acceleration.

\begin{figure*}
    \centering
    \begin{subfigure}[t]{0.55\textwidth}
        \centering
        \includegraphics[width=\linewidth]{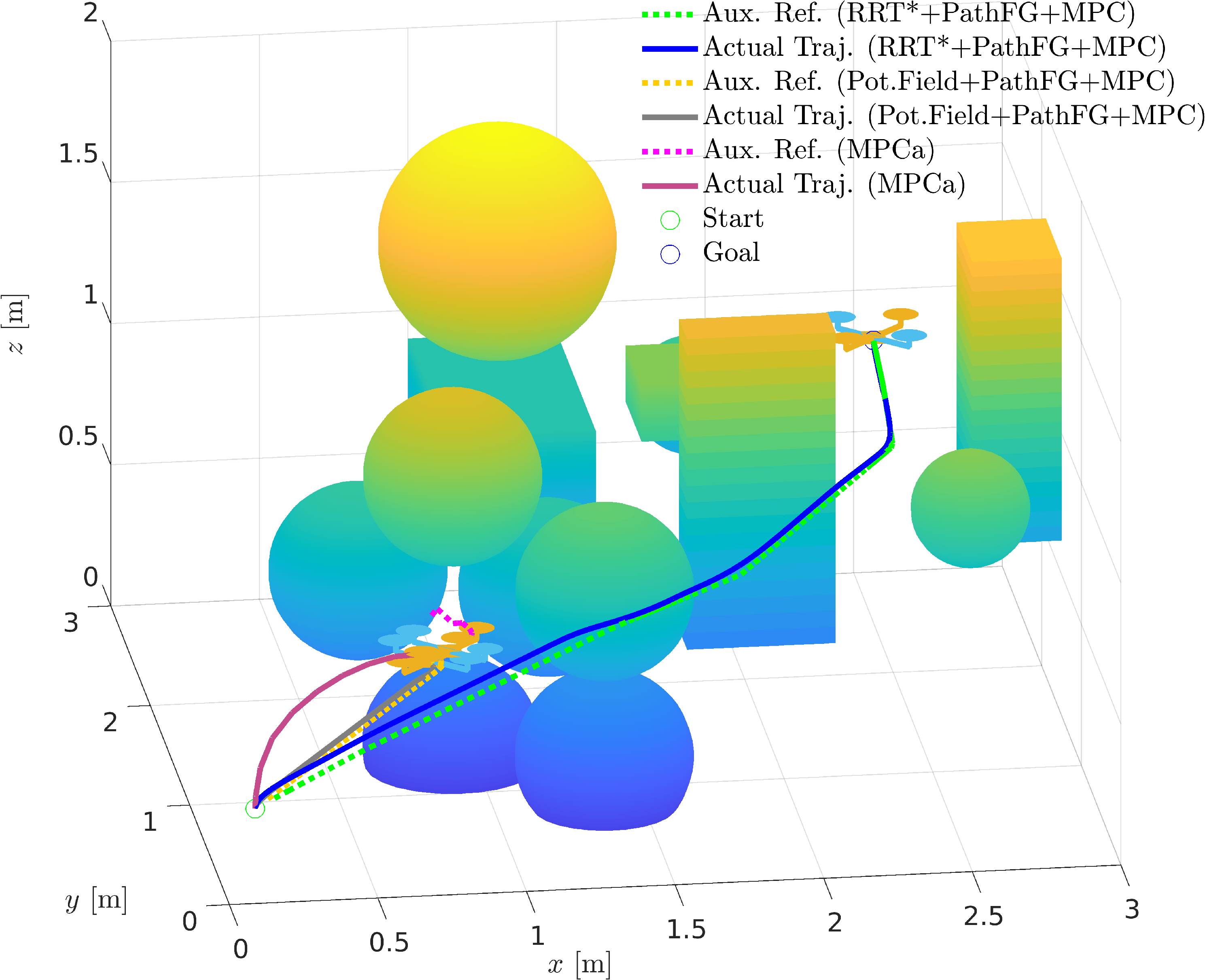}
        \caption{3D view of quadrotor navigation.}
        \label{fig:3dTrajRRT}
    \end{subfigure}\quad
    \begin{subfigure}[t]{0.43\textwidth}
        \centering
        \includegraphics[width=\linewidth]{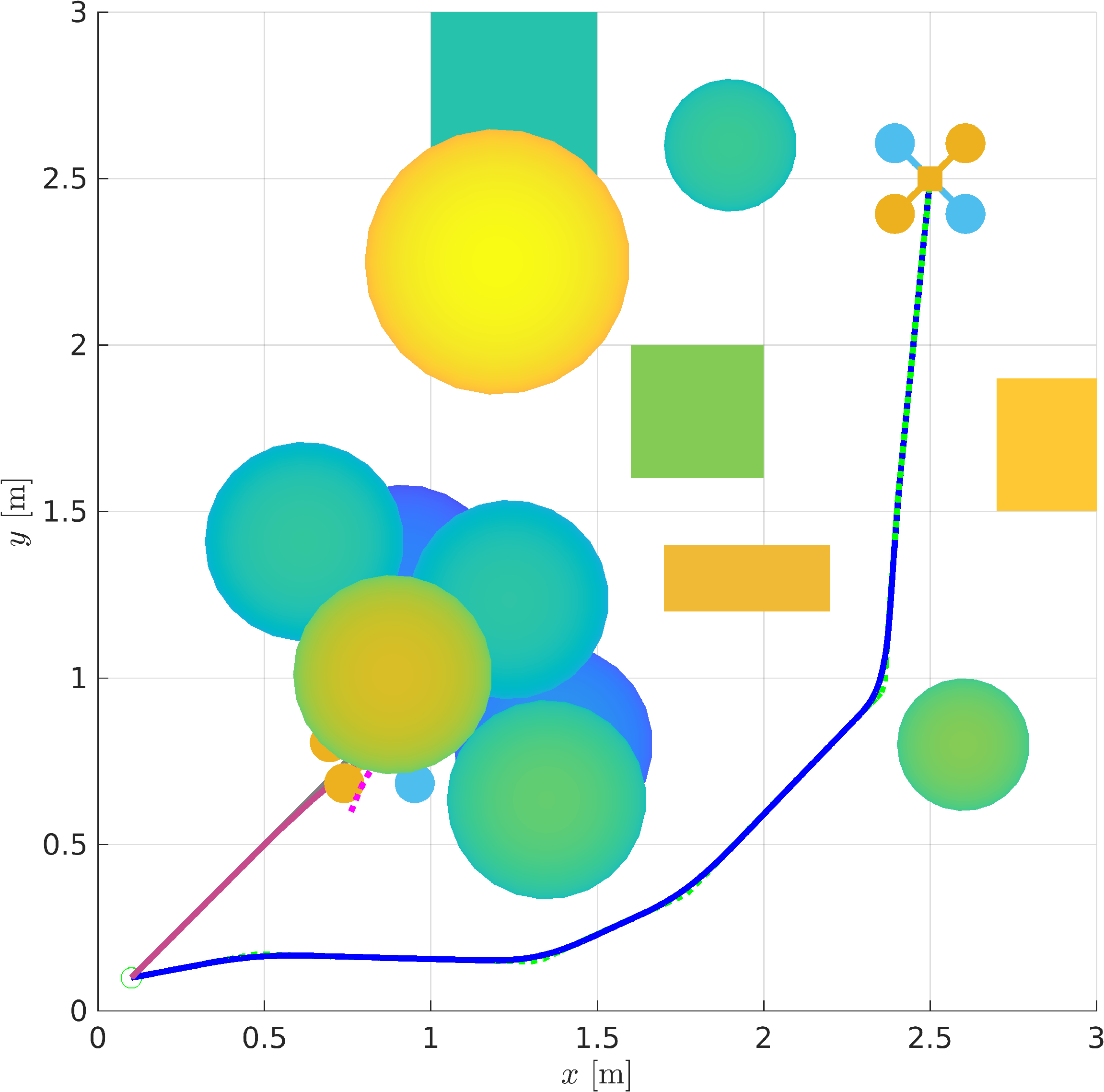}
        \caption{A 2D top-down view of Figure~\ref{fig:3dTrajRRT}.}
        \label{fig:topDownViewRRT}
    \end{subfigure}\\[1em]
    \bigskip
    \caption{Quadrotor navigation with obstacle avoidance using different methods: 1) RRT* planner + PathFG + MPC, 2) Potential Field planner + PathFG + MPC, and 3) MPCa (MPC with artificial reference). All methods use a short prediction horizon MPC with $N=5$. In the first case, the RRT* planner successfully generates a collision-free path (dotted line), from which PathFG dynamically selects a sequence of intermediate references (overlap with the path). These references are passed to the MPC controller to guide the system from the start position $(\SI{0.1}{m}, \SI{0.1}{m}, \SI{0.3}{m})$ to the goal $(\SI{2.5}{m}, \SI{2.5}{m}, \SI{1}{m})$. The resulting actual trajectory (solid line) demonstrates the capability of the PathFG framework to guide collision-free navigation through the obstacle-dense environment when a feasible path is provided. In contrast, in the second case, the Potential Field planner becomes trapped in the paraboloid-shaped obstacles and cannot provide a feasible path for PathFG to follow. Similarly, in the third case, MPCa also fails in the same environment due to the significant nonconvexities. Video is available at https://youtu.be/wpqDYO4r6Zg.}\label{fig:visualTraj}
\end{figure*}

Linearizing the dynamics of the system around $\bar x$ and $\bar u$ yields the following system. The state and control input are
\begin{subequations}
    \begin{equation}
        \Delta x \defas (\Delta p, \Delta\dot p, \Delta\psi) \defas x - \bar x
    \end{equation}
    and \noDisplaySkip
    \begin{equation}
        \Delta u \defas (\Delta T, \Delta\omega) \defas u - \bar u,
    \end{equation}
\end{subequations}
respectively. The linearized continuous-time state-space equation is then
\begin{equation}
    \Delta \dot x = A_\subtext{c}\+\Delta x + B_\subtext{c}\+\Delta u,
\end{equation}
where \noDisplaySkip
\begin{subequations}
    \begin{equation}
        A_\subtext{c} \defas \mat[
            0_3 & I_3 & 0_3\\
            0_3 & 0_3 & G  \\
            0_3 & 0_3 & 0_3
        ], \quad
        G \defas \mat[
            0  & g & 0\\
            -g & 0 & 0\\
            0  & 0 & 0
        ]
    \end{equation}
    and \noDisplaySkip
    \begin{equation}
        B_\subtext{c} \defas \mat[
            0_{5\times1} & 0_{5\times3}\\
            1/m          & 0_{1\times3}\\
            0_{3\times1} & I_3
        ]
    \end{equation}
\end{subequations}

To simplify notation, we will use these variables without the $\Delta$ prefix going forward. 

To obtain $x^\plus = Ax + Bu$, the linear continuous-time system is discretized using zero-order hold with a sampling time of $\SI{0.1}{s}$.

\subsubsection{Constraints} \label{sec:numericalConstraints}

As an example, we illustrate the procedure for linearizing a spherical obstacle\footnote{The projection onto box constraints is a trivial clipping.}, whose complements take the form
\begin{equation}
    \mathcal O_j^\complement = \set{x | \norm{\Xi\+x - o_j} \geq r_{o_j} \!+ r_\subtext{a}}.
\end{equation}
The centers of the spherical obstacles are denoted by $o_j \in \mathbb R^3$, and their corresponding radii by $r_{o_j}$, where $j \in \mathbb N_{[1,m]}$. The radius of the agent is denoted by $r_\subtext{a} \defas \SI{0.08}{m}$. In this case, the projection operator \ref{eq:projection} reduces to 
\begin{equation}
    \Pi_{\mathcal O_j}(x) = \Xi^\T \qty[o_j + \frac{\Xi\+x - o_j}{\norm{\Xi\+x - o_j}} \cdot (r_{o_j} \!+ r_a)]
\end{equation}
and the half-space approximation in \ref{eq:halfspace} becomes
\begin{equation} \label{eq:linSphereObstacle}
     H_{\mathcal O_j}(x) = \set{y | c_j(\Xi\+x)^\T y \leq d_j(\Xi\+x)},
\end{equation}
where \noDisplaySkip
\begin{equation} \label{eq:obstLinConstr} \begin{gathered}
    c_j(\Xi\+x) = \Xi^\T \wh{(o_j - \Xi\+x)}\\
    d_j(\Xi\+x) = \wh{(o_j - \Xi\+x)}^\T o_j - r_\subtext{a} - r_{o_j},
\end{gathered}\end{equation}
and the hat symbol denotes normalization: $\wh x \defas x/\norm{x}$.

By applying the formula \ref{eq:linSphereObstacle} to each predicted state $\xi_{i+1|k-1}^*$, we are able to obtain polyhedral approximations of state constraints across the entire predicted trajectory of OCP \ref{eq:linearOCP} as
\begin{equation}
    \mathcal X_{i|k} = \set{x | x_\subtext{min} \leq x \leq x_\subtext{max}} \cap \qty\bigg[\,\bigcap_{j=1}^m H_{\mathcal O_j}(\xi_{i+1|k-1}^*)],
\end{equation}
where $\thinmuskip=3mu plus 0.5mu x_\subtext{max} \defas (\SI{10}{m},\allowbreak \SI{10}{m},\allowbreak \SI{10}{m},\allowbreak \SI{1}{m/s},\allowbreak \SI{1}{m/s},\allowbreak \SI{1}{m/s},\allowbreak 0.2\pi,\allowbreak 0.2\pi,\allowbreak 0.2\pi)$ and $x_\subtext{min} \defas -x_\subtext{max}$ bound the system state.

The linear control input constraint set $\mathcal U$ is
\begin{equation}
    \mathcal U = [T_\subtext{min}{-}mg, T_\subtext{max}{-}mg] \times [\omega_\subtext{min}, \omega_\subtext{max}]^3,
\end{equation}
where $T_\subtext{max} \defas \SI{0.59}{N}$ and $T_\subtext{min} \defas \SI{0}{N}$ denote the maximum and minimum total thrusts, and $\omega_\subtext{max} \defas 0.5\pi$ and $\omega_\subtext{min} \defas -0.5\pi$ denote the maximum and minimum angular velocities, respectively.

\begin{figure}
    \centering
    \includegraphics[width=1.04\linewidth, center]{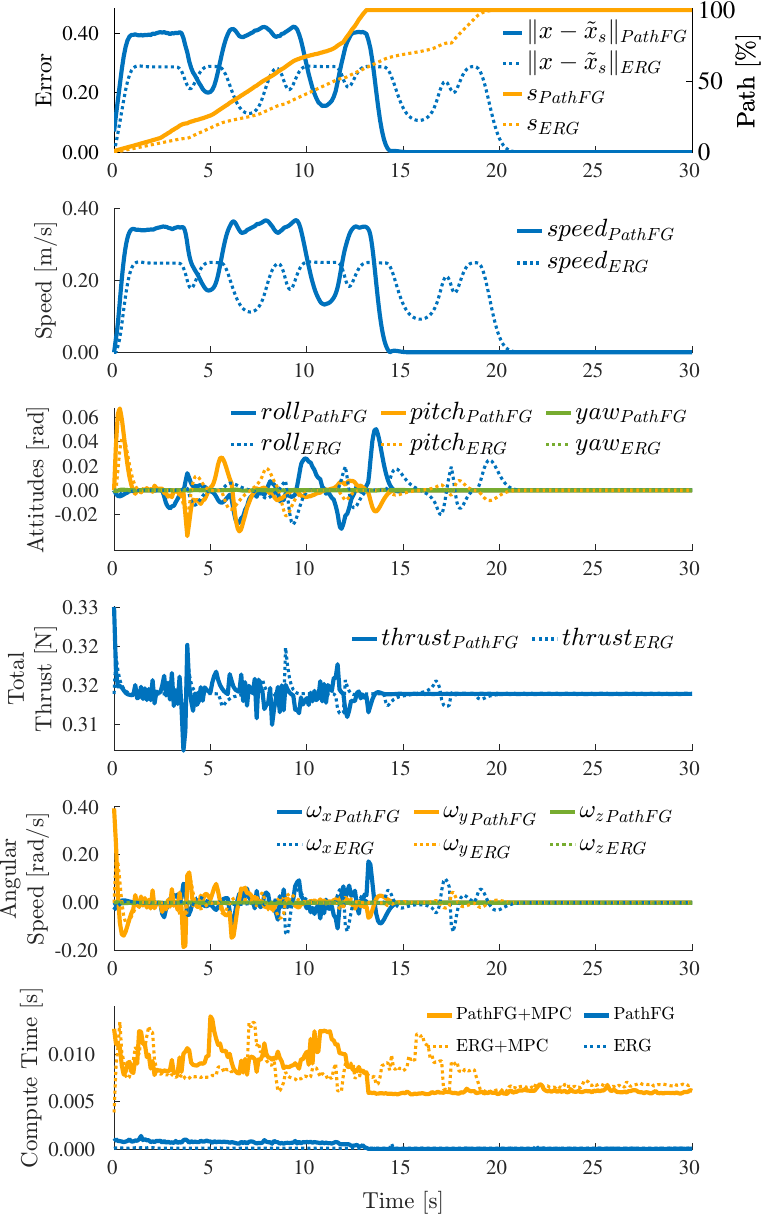}
    \caption{
    Comparison of closed-loop performance of PathFG and discrete ERG in the quadrotor simulation. Both methods use a short-horizon MPC with $N = 5$. Subfigure~1 shows the convergence of the error signal $e \defas \norm{x - \tilde x_s}$ to zero as the auxiliary reference $s$ reaches its target value of 100\%, with PathFG converging slightly faster than discrete ERG. All quadrotor states and control inputs remain within their prescribed bounds, confirming constraint satisfaction. The total computation time of both PathFG+MPC and ERG+MPC stays well below the \SI{0.1}{s} sampling period.}
    \label{fig:singlePredHorizon}
\end{figure}

\begin{figure}
    \centering
    \includegraphics[width=0.968\linewidth, center]{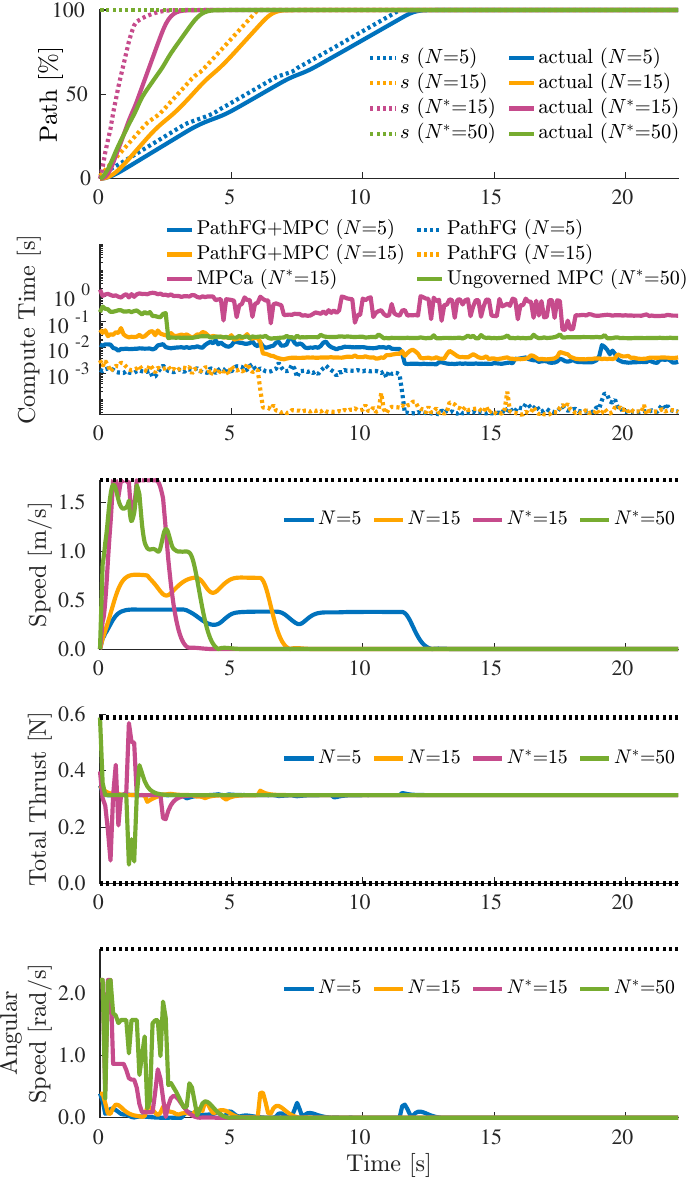}
    \caption{Evaluation of closed-loop performance in the quadrotor simulation under different control frameworks. The PathFG+MPC controller is tested with horizons $N = 5$ and $N = 15$, MPCa (provided with a good initial guess) is evaluated at $N^* = 15$, while ungoverned MPC is assessed at $N^* = 50$, where $N^*$ denotes the minimum horizon required to ensure feasibility of the OCP under the given constraints. Subfigures~1--2 illustrate the performance--computation trade-off: MPCa produces the most aggressive trajectories albeit at a higher computational cost. In contrast, PathFG+MPC maintains real-time performance with slower convergence. As shown in Subfigures~3--5, all configurations strictly adhere to the prescribed state and control bounds (black dotted lines), demonstrating safe navigation through the obstacle-dense environments.} \label{fig:compareDifferPredHorizon}
\end{figure}

\subsubsection{Primary Controller Implementation}

The cost matrices in \ref{eq:linearOCP} are chosen as
\begin{equation}\begin{gathered}
    \arraycolsep=0.2em
    Q = \mat[
        10 I_3 & 0       & 0      \\
        0      & 0.5 I_3 & 0      \\
        0      & 0       & 2.5 I_3
    ]
    \conjtext{and}
    R = \frac{I_4}{10}.
\end{gathered}\end{equation}
The terminal cost $P$ is obtained using the Discrete Algebraic Riccati Equation
\begin{equation}
    P = Q + A^\T PA - (A^\T PB)(R + B^\T PB)^{\!\inv}(B^\T PA),
\end{equation}
and the terminal feedback gain $K$ is the associated LQR gain
\begin{equation} \label{eq:termGain}
    K = (R + B^\T PB)^{\!\inv}(B^\T PA).
\end{equation}
Note that any terminal feedback gain $K$ satisfying the matrix inequality \ref{eq:LMI} is also admissible in lieu of the infinite-horizon LQR solution.

\subsubsection{PathFG Implementation} \label{sec:pathFGImpl}

The PathFG policy~\ref{eq:PathFGDef} is essentially a one-dimensional optimization problem that maximizes the auxiliary reference $s \in [0, 1]$ such that the pair $(\xi^*_N, s)$ lies within the invariant terminal set $\tilde \termset$. Since the search space of the optimization problem is just the interval $[0, 1]$, it can be efficiently solved using simple algorithms such as the bisection method.

The PathFG policy also involves constructing a terminal set, for which we apply formulation \ref{eq:termSetConstruct} with the following constraints (each row of the matrices will be used as a $c_i$ or $d_i$ in \ref{eq:linearConstraint}):
\begin{equation}\label{eq:matrixConstraints}
    \mat[
        \wh{(o_j-\Xi\+\tilde x_s)}^\T \Xi\\
        I\\
        \mathllap-I\\
        \mathllap-K_\subtext{thrust}\\
        K_\subtext{thrust}\\
        \mathllap-K_\omega\\
        K_\omega
    ] x \, \leq \mat[
        \wh{(o_j-\Xi\+\tilde x_s)}^\T o_j-r_\subtext{a}-r_{o_j}\\
        x_\subtext{max}\\
        -x_\subtext{min}\\
        T_\subtext{max}-mg-K_\subtext{thrust}\,\tilde x_s\\
        -T_\subtext{min}+mg+K_\subtext{thrust}\,\tilde x_s\\
        \omega_\subtext{max}-K_\omega\,\tilde x_s\\
        -\omega_\subtext{min}+K_\omega\, \tilde x_s
    ].
\end{equation}
Here, the obstacle constraints are linearized about the auxiliary reference $s$ using \ref{eq:linSphereObstacle}. The thrust gain $K_\subtext{thrust}$, defined as the first row of the terminal feedback gain $K$ in \ref{eq:termGain}, maps the state deviation to the additional thrust command. Similarly, the angular velocity gain $K_\omega$, defined as the corresponding rows of $K$, maps the state deviation to the $\omega_x$, $\omega_y$ or $\omega_z$ command.

\begin{figure}
    \centering
    \resizebox{\width}{1.0\height}{\includegraphics[width=1.0\linewidth, center]{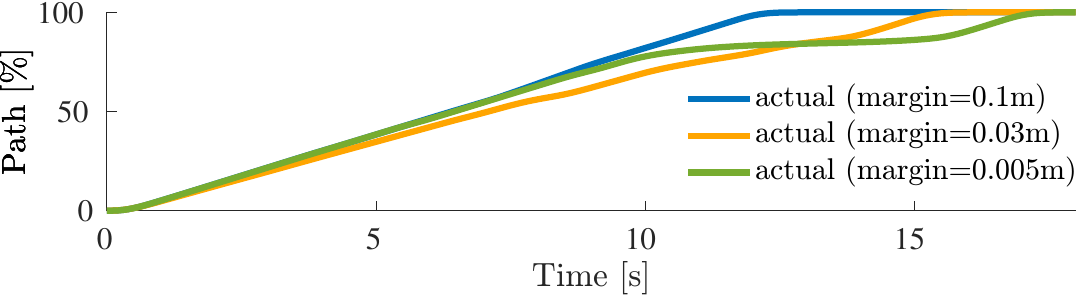}}    \caption{Comparison of PathFG performance using multiple RRT*-generated paths with different safety margins. Results show that as the path safety margin decreases, the actual system state converges more slowly toward the target.}
    \label{fig:PathFG_multiplePaths}
\end{figure}

\subsection{Simulation Results} 

In this section, we show the performance of the PathFG framework through MATLAB/Simulink simulations of Crazyflie navigation in obstacle-cluttered environments. The hardware configuration is identical to that used in Section~\ref{sec:Unicycle}.

Figure~\ref{fig:visualTraj} compares quadrotor navigation in an obstacle-cluttered environment using three control architectures: PathFG+MPC with an optimal path planner RRT* \cite{ref:Karaman2011}, PathFG+MPC with a suboptimal planner Potential Field \cite{ref:Khatib1985}, and MPCa implemented following \cite[Section 3.2.2]{Santos2023}. The results demonstrate that PathFG+MPC enables collision-free navigation when provided with a feasible path satisfying Assumption~\ref{as:path}. In contrast, both the Potential Field planner and MPCa struggle in this strongly nonconvex environment. The suboptimal performance of MPCa is attributable to the underlying gradient-based solver rather than the formulation itself; nevertheless, this limitation is practically significant since most MPC controllers rely on gradient-based methods, motivating the use of specialized planning algorithms and the PathFG.

Next, we evaluate the closed-loop performance of the proposed PathFG+MPC controller and compare it with the discrete ERG+MPC approach \cite{Convens2024}, as shown in Figure~\ref{fig:singlePredHorizon}. Both methods employ an RRT* planner to generate a feasible path. The results highlight the critical role of the governor in preserving MPC feasibility. In this strongly nonconvex environment, an ungoverned MPC requires a prediction horizon of at least $N = 50$ to remain feasible, whereas both PathFG+MPC and discrete ERG+MPC achieve safe navigation under the aggressive horizon constraint $N=5$. Both governors introduce minimal computational overhead, with discrete ERG being faster (${\sim}\SI{0.0002}{s}$ per update) than PathFG (${\sim}\SI{0.001}{s}$). This difference arises because discrete ERG avoids online optimization but requires careful, environment-dependent tuning of dynamic safety-margin scaling factors. In contrast, PathFG determines the auxiliary reference via a one-dimensional optimization, eliminating the need for such tuning at the cost of a modest increase in computation time.

Figure~\ref{fig:compareDifferPredHorizon} compares closed-loop performance across different control approaches, revealing a clear trade-off between performance and computation time. Note that MPCa, without a good initial guess, cannot find a feasible solution (Fig.~\ref{fig:visualTraj}), but when initialized with an RRT* path, it achieves the fastest convergence although at a higher computational cost. In contrast, PathFG+MPC operates in real time, but converges more slowly. It is also worth noting that the performance of PathFG+MPC can be partially recovered using a method proposed in prior feasibility governor study \cite{ref:Skibik2023}. By exploiting the design freedom in Assumption~\ref{as:linSysStability}, the terminal gain $K$ can be obtained using a ``less aggressive'' control strategy to enlarge the terminal set $\termset$, provided that the matrix inequality \ref{eq:LMI} remains satisfied. This reduces conservatism without modifying the terminal cost $P$.

Figure~\ref{fig:PathFG_multiplePaths} compares the impact of the path safety margin on PathFG+MPC performance ($N = 5$). Tighter safety margins lead to slower convergence, although convergence is still achieved in finite time. This is because a smaller margin shrinks the terminal set $\termset$, causing the auxiliary reference $s$ to progress more conservatively and, consequently, slowing the actual system states $x$. In practice, the path planner’s safety margin can be easily tuned based on available free space, enabling the agent to follow the path at higher speeds.

Lastly, Figure~\ref{fig:RRT_PotField} compares the performance of PathFG+MPC using paths generated by an optimal planner (RRT* \cite{ref:Karaman2011}) and a suboptimal planner (Potential Field \cite{ref:Khatib1985}) in an environment where both planners produce paths satisfying Assumption~\ref{as:path}. The results show that PathFG performance is similar with either planner, demonstrating robustness to the choice of path planner.

\begin{figure}
    \centering
    \includegraphics[width=1.0\linewidth, center]{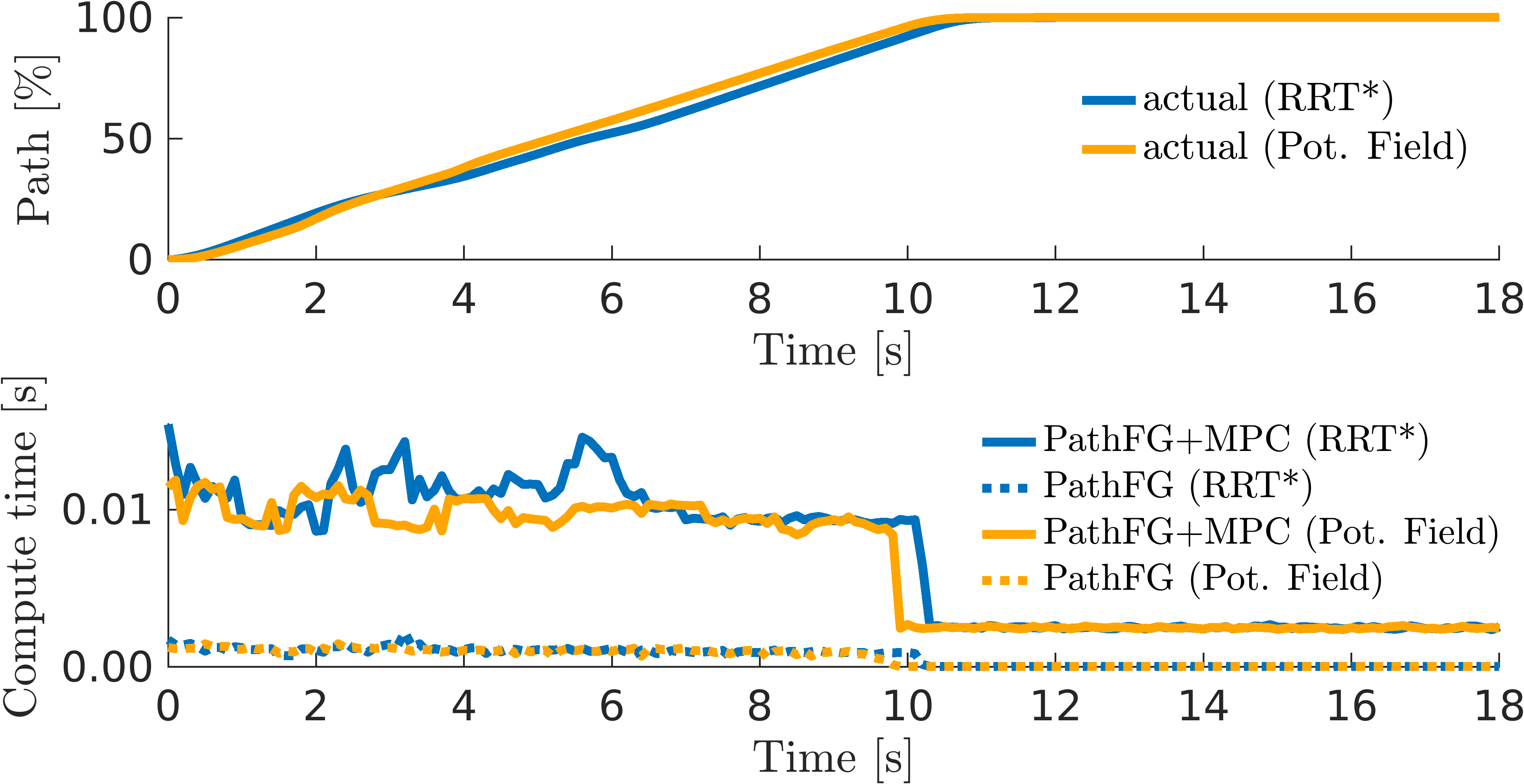}    \caption{Comparison of PathFG performance using an RRT*-generated path versus a Potential Field-generated path.}   \label{fig:RRT_PotField}
\end{figure}

\section{Conclusion}

This paper introduces the PathFG, a novel add-on unit that seamlessly integrates path planners with the nonlinear MPC. The PathFG is safe, asymptotically stable, converges in finite time and significantly expands the region of attraction of MPC. Rigorous theoretical analysis proves these properties, and quadrotor navigation simulations demonstrate its real-time performance. Future work will extend this method to handle moving obstacles, accommodate disturbances and model uncertainties, with experimental validation on real-world robotic platforms such as quadrotors.

\appendix

For notational simplicity, we omit the constant target reference $r$ in the appendix, therefore functions that depend on $r$ have one fewer or no parameters. We also assume a fixed feasible path $\mathcal P(x_0, r)$ throughout the proofs.

\subsection{Proof of Theorem~\ref{prop:safety}} \label{sec:proofOfsafety}

We will first prove that $(x_k, s_k) \in \tilde\Gamma$ for all $k \in \mathbb N$ via induction.

The base case is trivial as $(x_0, s_0) \in \tilde\Gamma$ follows directly from $x_0 \in \mathcal D_x$. Next, assume that $(x_k, s_k) \in \tilde\Gamma$. By the definition of $\tilde\Gamma$, an optimal solution to \ref{eq:OCP} at timestep $k$ exists and can be expressed as
\begin{equation}
    \zeta_k^* = (x_k, \xi_{1|k}^*, \dots, \xi_{N|k}^*, \mu_{0|k}^*, \mu_{1|k}^*, \dots, \mu_{N-1|k}^*).
\end{equation}
Moreover, since $(x_k, s_k) \in \tilde\Gamma$ implies $(\xi_{N|k}^*, s_k) \in \tilde\termset$, it follows that $s_{k+1} = s_k$ is a valid solution to the PathFG policy $\smash[b]{g(\xi_{N|k}^*)}$ defined in \ref{eq:PathFGDef} and that $s_{k+1} = \tilde g(x_k, s_k)$ exists. Since $(\xi_{N \mid k}^*, s_{k+1}) \in \tilde\termset$ according to \ref{eq:PathFGDef}, we proceed to construct the following sequences for timestep $k + 1$:
\begin{equation}\begin{multlined}
    \zeta_{k+1} = (\xi_{1|k}^*, \dots, \xi_{N|k}^*, f(\xi_{N|k}^*, \tilde\kappa_\subtext{T}(\xi_{N|k}^*, s_{k+1})),\\
    \mu_{1|k}^*, \dots, \mu_{N-1|k}^*, \tilde\kappa_\subtext{T}(\xi_{N|k}^*, s_{k+1})),
\end{multlined}\end{equation}
where $\tilde\kappa_\subtext{T}$ is the terminal control law defined in \ref{eq:termCtrlLaw}. As the terminal set is invariant by Assumption~\ref{as:termInvariant}, the sequence $\zeta_{k+1}$ is a feasible solution to the OCP \ref{eq:OCP} at timestep $k + 1$. Since $x_{k+1} = \xi_{1|k}^*$, we have shown that $(x_{k+1}, s_{k+1}) \in \tilde\Gamma$.
    
Lastly, $x_k \in \mathcal X$ and $u_k \in \mathcal U$ follow directly from $(x_k, s_k) \in \tilde\Gamma$ for all $k \in \mathbb N$.
    
\subsection{Proof of Theorem~\ref{prop:AS}} \label{sec:proofOfAS}

The proof of Theorem~\ref{prop:AS} is broken down into several lemmas. The central idea of the proof is to demonstrate that, if the change in the auxiliary reference $s_k$ remains small over a sufficiently long period, the state of system \ref{eq:noPathFGDynamics} will become ``close enough'' to the reference state $\tilde x_s$ under the action of the MPC controller. This property allows $s_k$ to jump towards the target by some minimal distance along the planned path while maintaining the OCP feasibility. Hence, the change in $s_k$ will not become infinitesimally small, establishing the finite-time convergence of $s_k$ to $1$.

First, we demonstrate Input-to-State Stability (ISS) of \ref{eq:noPathFGDynamics} under the MPC control law with a varying auxiliary reference $s$, a result that will be leveraged in the subsequent proofs.

\begin{lemma}[ISS] \label{prop:ISS}
    Given Assumptions~\ref{as:system,as:nonlinSysStability,as:path}, an initial state $x_0 \in \mathcal D_x$, and the system $x_{k+1} = \tilde f(x_k, s_k)$, the error signal
    \begin{equation} \label{eq:error}
        e_k \defas x_k - \tilde x_{s_k}
    \end{equation}
    is ISS with respect to the input $\Delta s_k \defas s_{k+1} - s_k$, i.e., there exist $\beta \in \mathcal{KL}$ and $\gamma \in \mathcal K$ such that
    \begin{equation} \label{eq:ISS}
        \norm{x_k - \tilde x_{s_k}} \leq \beta(\norm{x_0 - \tilde x_{s_0}}, k) + \gamma\qty\Big(\+\sup_{j\in\mathbb N} \abs{\Delta s_j})
    \end{equation}
    for all $k \in \mathbb N$ and feasible initial conditions $(x_0, s_0) \in \tilde\Gamma$. Moreover, $2\gamma$ is an asymptotic gain, i.e., 
    \begin{equation} \label{eq:asymptoticGain}
        \limsup_{k\to\infty} \norm{x_k - \tilde x_{s_k}} \leq 2\+\gamma\qty\Big(\limsup_{k\to\infty} \abs{\Delta s_k}).
    \end{equation}
\end{lemma}

\begin{proof}
    For any $s \in [0, 1]$, the optimal cost function of the MPC feedback law $\tilde J(\phdot, s)$, as defined in \ref{eq:tildeJ}, is a Lyapunov function for the closed-loop system \cite{ref:Mayne2000}, i.e., there exist $\alpha, \alpha_\ell, \alpha_u \in \mathcal K_\infty$ such that for all $(x, s) \in \tilde\Gamma$,
    \begin{equation} \label{eq:ISSLyapunov} \begin{gathered}
        \tilde J(\tilde f(x, s), s) - \tilde J(x, s) \leq - \alpha(\norm{x - \tilde x_s})\\
        \alpha_\ell(\norm{x - \tilde x_s}) \leq \tilde J(x, s) \leq \alpha_u(\norm{x - \tilde x_s}).
    \end{gathered}\end{equation}
    Furthermore, under Assumption \ref{as:costContinuity}, $\tilde J$ is uniformly continuous and thus, there exists $\sigma\_[2]x, \sigma\_[2]s \in \mathcal K$ such that
    \begin{equation}
        \abs*{\tilde J(x^\plus, s^\plus) - \tilde J(x, s)} \leq \sigma\_[2]x(\norm*{x^\plus - x}) + \sigma\_[2]s(\abs*{s^\plus - s}).
    \end{equation}
    Hence, for any $(x, s) \in \tilde\Gamma$, $x^\plus = \tilde f(x, s)$ and $(x^\plus, s^\plus) \in \tilde\Gamma$, we have
    \begin{equation}\begin{aligned}
        \Delta \tilde J &\defas \tilde J(x^\plus, s^\plus) - \tilde J(x, s)\\
            &= \tilde J(x^\plus, s) - \tilde J(x, s) + \tilde J(x^\plus, s^\plus) - \tilde J(x^\plus, s)\\ 
            &\leq -\alpha(\norm{x - \tilde x_s}) + \abs*{\tilde J(x^\plus, s^\plus) - \tilde J(x^\plus, s)}\\
            &\leq -\alpha(\norm{x - \tilde x_s}) + \sigma\_[2]s(\abs*{s^\plus - s}).
    \end{aligned}\end{equation}
    which demonstrates ISS \cite[Lemma~3.5]{ref:Jiang2001}. Now that we have proven \ref{eq:ISS}, the property of $a + b \leq \max\qty{2a, 2b}$ for any $a, b \in \mathbb R$ implies
    \begin{equation}\begin{multlined}
        \norm{x_k - \tilde x_{s_k}} \leq \beta(\norm{x_0 - \tilde x_{s_0}}, k) + \gamma\qty\Big(\+\sup_{j\in\mathbb N} \abs{\Delta s_j})\\
        \leq \max\qty\Big{2\+\beta(\norm{x_0 - \tilde x_{s_0}}, k), 2\+\gamma\qty\Big(\+\sup_{j\in\mathbb N} \abs{\Delta s_j})},
    \end{multlined}\end{equation}
    where \cite[Lemma~3.8]{ref:Jiang2001} can be applied to prove that $2\gamma$ is an asymptotic gain.
\end{proof}

With the system's ISS property established, we proceed to show that the asymptotic distance between the optimal final predicated state $\xi^*_{N|k}$ and the reference state $\tilde x_{s_k}$ decreases as the change in auxiliary reference $\Delta s_k$ decreases.

\begin{lemma} \label{prop:xiConvrgToXBar}
    Given Assumptions~\ref{as:system,as:nonlinSysStability,as:path}, and an initial state $x_0 \in \mathcal D_x$, there exists $\bar\gamma \in \mathcal K$ such that
    \begin{equation}
        \limsup_{k\to\infty} \norm*{\xi_{N|k}^* - \tilde x_{s_k}} \leq \bar \gamma\qty\Big(\limsup_{k\to\infty} \abs{\Delta s_k}).
    \end{equation}
\end{lemma}

\begin{proof}
    From \ref{eq:optimalCost} and \ref{eq:ISSLyapunov}, and noting that $\ell$ is positive semi-definite, we have
    \begin{equation}
        V(\tilde\xi_N^*(x, s), \p(s)) \leq \tilde J(x, s) \leq \alpha_u(\norm{x - \tilde x_s})
    \end{equation}
    for all $(x, s) \in \tilde\Gamma$. Using this bound together with \ref{eq:asymptoticGain}, there exists an $\iota \in \mathcal K$ from Assumption~\ref{as:nolinTermCost} such that
    \begin{equation}\begin{aligned}
        \limsup_{k\to\infty} \iota(\norm*{\tilde\xi_{N|k}^* - \tilde x_{s_k}}) &\leq \limsup_{k\to\infty} V(\tilde\xi_{N|k}^*, \p(s_k)) \\
            &\leq \limsup_{k\to\infty} \tilde J(x_k, s_k)\\
            &\leq \alpha_u\qty\Big(\limsup_{k\to\infty} \norm{x_k - \tilde x_{s_k}})\\
            &\leq [\alpha_u \circ 2\gamma]\qty\Big(\limsup_{k\to\infty} \abs{\Delta s_k}),
    \end{aligned}\end{equation}
    which implies \noDisplaySkip
    \begin{equation}
        \limsup_{k\to\infty} \norm*{\tilde\xi_{N|k}^* - \tilde x_{s_k}} \leq [\iota^\inv \circ \alpha_u \circ 2\gamma]\qty\Big(\limsup_{k\to\infty} \abs{\Delta s_k}).
    \end{equation}
    The results follows by defining $\bar\gamma = \iota^\inv \circ \alpha_u \circ 2\gamma$.
\end{proof}

The following Lemmas~\ref{prop:largeRefChange,prop:largeRefChangeEta,prop:finiteTimeConvrgOfS} build up to the result that the sequence of auxiliary references $\qty{s_k}_{k=0}^\infty$ produced by PathFG converges to $s = 1$ in finite time.

We first show that if the optimal final predicated state $\xi_N^*$ is sufficiently close to the equilibrium state $\tilde x_s$, the auxiliary reference $s$ is able to move some minimal distance $\alpha$ towards the final reference $1$.

We first need to define the set of equilibrium states on the path $\mathcal P(x_0)$ as
\begin{equation}
    \Sigma = \set{(\tilde x_s, \p(s)) | s \in [0, 1]},
\end{equation}
and the $\delta$-neighborhood around $\Sigma$ with respect to $x$ as
\begin{equation}
    \mathcal B_\delta(\Sigma) \defas \set{(x, \p(s)) | \norm{x - \tilde x_s} \leq \delta,\, s \in [0, 1]}.
\end{equation}

\begin{lemma} \label{prop:largeRefChange}
    Given Assumptions~\ref{as:system,as:nonlinSysStability,as:path}, and an initial state $x_0 \in \mathcal D_x$, there exists a $\delta > 0$ and $\alpha > 0$ such that for all $(\xi_{N|k}^*,\allowbreak \p(s_k)) \in \mathcal B_\delta(\Sigma)$,
    \begin{equation}\begin{aligned}
        & s_{k+1} - s_k \geq \alpha && \text{if } 1 - s_k > \alpha\\
        & s_{k+1} = 1 && \text{if } 1 - s_k \leq \alpha.
    \end{aligned}\end{equation}
\end{lemma}

\begin{proof}
    First, we will prove that $\Sigma$ is compact. Consider the function $h\colon \mathcal P(x_0) \to \Sigma$ such that $h(\p(s)) = (\tilde x_s, \p(s))$. Since $\bar x_{(\cdot)}$ is continuous by Assumption~\ref{as:system}, so is $h$. Because $\mathcal P(x_0)$ is compact, $\Sigma = h(\mathcal P(x_0))$ is also compact.

    We claim that since $\Sigma$ is compact, $\Int\termset$ is open, and $\Sigma \subseteq \Int\termset$ (by Assumption~\ref{as:equilInTermSet}), there exists a minimal distance $\vartheta > 0$ between $\Sigma$ and the boundary of $\termset$, i.e., 
    \begin{equation} \label{eq:ballSubsetIntT}
        \mathcal B_\vartheta\qty\big((\tilde x_s, \p(s))) \subseteq \Int\termset, \quad \forall (\tilde x_s, \p(s)) \in \Sigma.
    \end{equation}
               
    Assume that the claim is false. Then there are sequences $\qty{a_n}_{n=0}^\infty \subseteq \Sigma$ and $\qty{\delta_n}_{n=0}^\infty \subseteq \mathbb R_{>0}$ such that $\delta_n \to 0$ as $n \to \infty$ and $\mathcal B_{\delta_n}\!(a_n) \not\subseteq \Int\termset$ for all $n \in \mathbb N$. Due to the compactness of $\Sigma$, there exists a subsequence $\qty{a_{n_k}}_{k=0}^\infty$ such that $a_{n_k} \?\to a$ as $k \to \infty$ for some $a \in \Sigma$. Since $\Int\termset$ is open, there is some $\delta^* > 0$ such that $\mathcal B_{\delta^*\!}(a) \subseteq \Int\termset$. Choose $k$ sufficiently large so that $\delta_{n_k} \!< \delta^*\!/2$ and $\norm{a_{n_k} \?- a} < \delta^*\!/2$, which implies $\mathcal B_{\delta_{n_k}}\!(a_{n_k}) \subseteq \mathcal B_{\delta^*\!}(a) \subseteq \Int\termset$, contradicting our assumption that $\mathcal B_{\delta_n}\!(a_n) \not\subseteq \Int\termset$ for all $n \in \mathbb N$. Hence, we have proved the existence of $\vartheta$
                      
    Next, we claim that by choosing $\delta = \vartheta/2$ and $\rho = \vartheta/2$, there is a minimal distance of $\rho$ between $\mathcal B_\delta(\Sigma)$ and the boundary of $\termset$ with respect to $\p(s)$, i.e., if $(\xi_{N|k}^*, \p(s_k)) \in \mathcal B_\delta(\Sigma)$, then
    \begin{equation} \label{eq:minDistOfRhoFixXi}
        (\xi_{N|k}^*, \p(s)) \in \termset, \quad \forall \p(s) \in \mathcal B_\rho(\p(s_k)).
    \end{equation}
    This follows from the fact that if $\p(s) \in \mathcal B_\rho(\p(s_k))$, then
    \begin{equation}\begin{aligned}
            & \norm*{(\xi_{N|k}^*, \p(s)) - (\tilde x_{s_k}^*, \p(s_k))}\\
        \leq{} & \begin{multlined}[t]
            \norm*{(\xi_{N|k}^*, \p(s)) - (\xi_{N|k}^*, \p(s_k))}\\
            + \norm*{(\xi_{N|k}^*, \p(s_k)) - (\tilde x_{s_k}, \p(s_k))}
        \end{multlined}\\
        ={} & \norm{\p(s) - \p(s_k)} + \norm*{\xi_{N|k}^* - \tilde x_{s_k}}\\
        \leq{} & \rho + \delta = \tfrac{\vartheta}{2} + \tfrac{\vartheta}{2} =\vartheta,
    \end{aligned}\end{equation}
    which means that
    \begin{equation}
         (\xi_{N|k}^*, \p(s)) \in \mathcal B_\vartheta\qty\big((\tilde x_{s_k}, \p(s_k))) \subseteq \termset
    \end{equation}
    due to \ref{eq:ballSubsetIntT} as claimed. 

    Next, we prove that there exists an $\alpha > 0$ such that if $(\xi_{N|k}^*, \p(s_k)) \in \mathcal B_\delta(\Sigma)$, then 
    \begin{equation} \label{eq:intervalOfRadiusAlpha}
        (\xi_{N|k}^*, \p(s)) \in \termset, \quad \forall s \in \mathcal B_\alpha(s_k) \cap [0, 1].
    \end{equation}
    Because $\p$ is continuous on a closed interval $[0, 1]$, it is uniformly continuous. Therefore, there exists a uniform $\alpha = \alpha(\delta) > 0$ (that works for all $s_k$) such that if $\abs{s - s_k} \leq \alpha$, then $\norm{\p(s) - \p(s_k)} \leq \rho$, and so $(\xi_{N|k}^*, \p(s)) \in \termset$ according to \ref{eq:minDistOfRhoFixXi}. 
    
    Finally, we prove the main statement of the lemma. If $1 - s_k > \alpha$, then $s_{k+1} \geq s_k + \alpha$ because \ref{eq:intervalOfRadiusAlpha} implies $(\xi_{N|k}^*, \p(s_k + \alpha)) \in \termset$ and  
    \begin{equation*}
        s_{k+1} = g(\xi_{N|k}^*) = \max\set{s | (\xi_{N|k}^*, \p(s)) \in \termset} \geq s_k + \alpha.
    \end{equation*}
    On the other hand, if $1 - s_k \leq \alpha$, then $s_{k+1} = 1$ since $(\xi_{N|k}^*,\allowbreak \p(1)) \in \termset$ by \ref{eq:intervalOfRadiusAlpha}.
\end{proof}

The following lemma proves that the conditions for Lemma~\ref{prop:largeRefChange} will eventually be satisfied given a finite, sufficiently long wait. This ensures that the auxiliary reference $s$ will always be able to jump at least some minimal distance towards the final reference $1$.

\begin{lemma} \label{prop:largeRefChangeEta}
    Let Assumptions~\ref{as:system,as:nonlinSysStability,as:path} hold, and let $x_0 \in \mathcal D_x$ and $\eta \in (0, \min\qty{\alpha, \bar\gamma^\inv(\delta)})$. For any timestep $k_0$, there exists a finite time $k_i > k_0$ such that
    \begin{equation} \label{eq:largeRefChangeEta}
        s_{k_i+1} - s_{k_i} \geq \eta \conjtext{or} s_{k_i+1} = 1.
    \end{equation}
\end{lemma}

\begin{proof}
    Assume there is no such $k_i$. Then,
    \begin{equation}
       \abs{\Delta s_k} \defas s_{k+1} - s_k < \eta < \bar\gamma^\inv(\delta), \quad \forall k > k_0,
    \end{equation}
    and so by applying $\limsup$ we have
    \begin{equation}
        \limsup_{k\to\infty} \abs{\Delta s_k} < \limsup_{k\to\infty} \bar\gamma^\inv(\delta) = \bar\gamma^\inv(\delta),
    \end{equation}
    which we can then apply $\bar\gamma$ to get
    \begin{equation}
        \bar\gamma\qty\Big(\limsup_{k\to\infty} \abs{\Delta s_k}) < [\bar\gamma \circ \bar\gamma^\inv](\delta) = \delta.
    \end{equation}
    Then by Lemma~\ref{prop:xiConvrgToXBar}, we have
    \begin{equation}    
        \limsup_{k\to\infty} \norm*{\xi_{N|k}^* - \tilde x_{s_k}} < \bar\gamma\qty\Big(\limsup_{k\to\infty} \abs{\Delta s_k}) < \delta.
    \end{equation}
    Therefore, there exists a finite $k_i$ such that $\xi_{N|k_i}^*$ is close enough to $\tilde x_{s_{k_i}}$, i.e., $(\xi_{N|k_i}^*, \p(s_{k_i})) \in \mathcal B_\delta(\Sigma)$. By virtue of Lemma~\ref{prop:largeRefChange}, either
    \begin{equation}
        s_{k_i+1} - s_{k_i} \geq \alpha > \eta \conjtext{or} s_{k_i+1} = 1,
    \end{equation}
    concluding the proof.
\end{proof}

Next, we prove the existence of sufficiently many such ``large'' jumps in the auxiliary reference $s$, which guarantees the finite-time convergence of $s$ to $1$.

\begin{lemma}[Finite Time Convergence of $s_k$] \label{prop:finiteTimeConvrgOfS}
    Given Assumptions~\ref{as:system,as:nonlinSysStability,as:path}, and an initial condition $x_0 \in \mathcal D_x$, there is a $k_M \geq 0$ such that $s_k = 1$ for all $k \geq k_M + 1$.
\end{lemma}

\begin{proof}
    Using $k_0 = 0$ in Lemma~\ref{prop:largeRefChangeEta}, there exists a finite time $k_1$ such that \ref{eq:largeRefChangeEta} is satisfied when $i = 1$. Using $k_0 = k_1$, there exists a finite time $k_2 > k_1$ such that \ref{eq:largeRefChangeEta} is satisfied when $i = 2$. Repeat this until there is an increasing sequence of $M \defas \lceil1/\eta\rceil$ timesteps $k_1 < k_2 < \dots < k_M$.

    From \ref{eq:largeRefChangeEta}, there are two possibilities for each $i$: $s_{k_i+1} - s_{k_i} \geq \eta$ or $s_{k_i+1} = 1$. If there exists an $i$ such that $s_{k_i+1} = 1$, then the desired result is immediately obtained. Therefore, the worst-case scenario occurs when
    \begin{equation}
        s_{k_i+1} - s_{k_i} \geq \eta, \quad \forall i \in \mathbb N_{[1,M]}.
    \end{equation}
    Then \noDisplaySkip
    \begin{equation}\begin{aligned}
        s_{k_M+1} &= s_{k_M} + (s_{k_M+1} - s_{k_M})\\
            &\geq s_{k_M} + \eta \geq s_{k_{M-1}+1} + \eta\\
            &\geq s_{k_{M-1}} + 2\eta \geq s_{k_{M-2}+1} + 2\eta\\
            &\quad\vdots\\
            &\geq s_{k_1} + M\eta = s_{k_1} + \bigl\lceil\tfrac{1}{\eta}\bigr\rceil \eta\\
            &\geq s_{k_1} + 1 \geq 1.
    \end{aligned}\end{equation}
    The second inequality holds because $s$ cannot backtrack, so
    \begin{equation}\begin{aligned}
        k_M > k_{M-1} &\implies k_M \geq k_{M-1} + 1\\
            &\implies s_{k_M} \geq s_{k_{M-1}+1}.
    \end{aligned}\end{equation}
    Since $s_{k_M+1}$ must be in the range $[0, 1]$, it follows that $s_{k_M+1} = 1$.
\end{proof}

Combining the input-to-state stability of \ref{eq:noPathFGDynamics} (Lemma~\ref{prop:ISS}) and the finite-time convergence of $s$ to $1$ (Lemma~\ref{prop:finiteTimeConvrgOfS}), Theorem~\ref{prop:AS} directly follows.

\subsection{Proof of Lemma~\ref{prop:linProperties}} \label{sec:proofOflinProperties}

First, we verify that Assumption~\ref{as:linStageCost} satisfies Assumption~\ref{as:nonlinStageCost}. The proofs for the uniform continuity of $\ell$ and the condition $\ell(\bar x_r, \bar u_r, r) = 0$ are trivial. To show that $\ell$ is bounded below by a $\mathcal K_\infty$-function $\gamma$, we choose $\gamma$ as $\lambda_\subtext{min} \norm{x - \bar x_r}^2$, where $\lambda_\subtext{min}$ is the smallest eigenvalue of $Q$. Since $Q$ is real, positive definite and symmetric, $\lambda_\subtext{min}$ is guaranteed to be real.

Next, we prove that Assumption~\ref{as:decreasingLyapunov} satisfies Assumption~\ref{as:nolinTermCost}. The uniform continuity of $V$, and the fact that $V(\bar x_r, r) = 0$ and $V(x, r) \geq 0$ for all $(x, r) \in \termset$ are also trivial. The terminal control law $\kappa_\subtext{T}$ can be chosen as in \ref{eq:lintermlaw}. Thus, the inequality in \ref{eq:LMI} implies
\begin{equation}\begin{aligned}
        & {-}\delta x^\T(K^\T RK + Q)\,\delta x\\
    ={} & -(K\+\delta x)^\T* R\+(K\+\delta x) - \delta x^\T Q\,\delta x\\
    ={} & {-}\ell(x, \kappa_\subtext{T}(x, r), r)\\
    \geq{} & \delta x^\T(A - BK)^\T* P\+(A - BK)\,\delta x - \delta x^\T P\,\delta x\\
    ={} & \delta x^{\plus\T}\!P\,\delta x^\plus - \delta x^\T P\,\delta x\\
    ={} & V(x^\plus, r) - V(x, r),
\end{aligned}\end{equation}
where $\delta x \defas x - \bar x_r$, $x^\plus \defas Ax + B\kappa_\subtext{T}(x, r)$ and $\delta x^\plus \defas x^\plus - \bar x_r$, which proves the satisfaction of \ref{eq:termCostIneq} in Assumption~\ref{as:nolinTermCost}. 

Assumption~\ref{as:costContinuity} is satisfied by the results in \cite{ref:Bemporad2002}.

\subsection{Proof of Lemma~\ref{prop:linObsProperties}} \label{sec:proofOflinObsProperties}

For notational simplicity, we will let $\Pi_x = \Pi_{\mathcal O_j}(x)$.

We prove Lemma~\ref{prop:linObsSafe} by contradiction. Assume that the linearized obstacle constraints \ref{eq:halfspace} overlap with the obstacles, i.e., $H_{\mathcal O_j}(x) \cap \mathcal O_j \neq \emptyset$. The obstacle itself can be either open or closed; we assume the inflated obstacle $\mathcal O_j$ to be open. It follows that $\mathcal O_j$ and $\Int[H_{\mathcal O_j}(x)]$ must also intersect. Thus, there exists an $y \in \mathcal O_j$ such that $(x - \Pi_x)^\T\, (y - \Pi_x) > 0$. Since $\closure{\mathcal O_j}$, the closure of the obstacle, is convex, and both $y, \Pi_x \in \closure{\mathcal O_j}$, we have \noDisplaySkip*
\begin{equation} \label{eq:obstacleConvex}
    \Pi_x + t\,(y - \Pi_x) \in \closure{\mathcal O_j}, \quad \forall t \in [0, 1].
\end{equation}
Choose a $t \in (0, 1)$ such that
\begin{equation}
    t < \frac{2\,(x-\Pi_x)^\T\,(y-\Pi_x)}{\norm{y-\Pi_x}^2}.
\end{equation}
Then, \noDisplaySkip**
\begin{equation}\begin{aligned}
        & t^2\,\norm{y - \Pi_x}^2 < 2t\,(x - \Pi_x)^\T\, (y - \Pi_x)\\
    \implies& \begin{multlined}[t]
        \norm{x - \Pi_x}^2 - 2t\,(x - \Pi_x)^\T\, (y - \Pi_x)\\
        + t^2\,\norm{y - \Pi_x}^2 < \norm{x - \Pi_x}^2
    \end{multlined}\\
    \implies& \norm{(x - \Pi_x) - t\,(y - \Pi_x)}^2 < \norm{x - \Pi_x}^2\\
    \implies& \norm{x - [\Pi_x + t\,(y - \Pi_x)]} < \norm{x - \Pi_x}.
\end{aligned}\end{equation}
This means that the point $\Pi_x + t\,(y - \Pi_x)$ \dash which is in $\closure{\mathcal O_j}$ by \ref{eq:obstacleConvex} \dash is closer to $x$ than $\Pi_x$, contradicting the definition of $\Pi_x$ in \ref{eq:projection}.

To prove Lemma~\ref{prop:linObsInterior}, note that $x \in \Int[H_{\mathcal O_j}(x)]$ because
\begin{equation}
    (x - \Pi_x)^\T\, (x - \Pi_x) = \norm{x - \Pi_x}^2 > 0.
\end{equation}

\subsection{Proof of Lemma~\ref{prop:linTermSet}} \label{sec:proofOflinTermSet}

First, we show that the terminal set $\termset$ constructed in \ref{eq:termSetConstruct} satisfies Assumption~\ref{as:termInvariant}, i.e., $\termset$ is positively invariant and constraint admissible under the terminal dynamics. We utilize a result from \cite[\S\,\emph{Computation of the Threshold Value}]{ref:Nicotra2018} which guarantees constraint satisfaction if the Lyapunov value $V(x, r)$ never exceeds the Lyapunov threshold value $\Lambda(r)$ as in \ref{eq:lyapunovThreshold}, i.e.,
\begin{equation}
    V(x, r) \leq \Lambda(r) \implies c(r)^\T x \leq d(r).
\end{equation}
Using this result, we see that $\Delta(x, r) \leq 0$ implies satisfaction of all the linearized constraints \ref{eq:linearConstraint}, which, in turn, ensures constraint admissibility by Lemma \ref{prop:linObsSafe}, i.e.,
\begin{equation}
    \Delta(x, r) \leq 0 \implies x \in \mathcal X \text{ and } \kappa_\subtext{T}(x, r) \in \mathcal U.
\end{equation}
Moreover, since $V$ is nonincreasing over time under the terminal dynamics by Assumption~\ref{as:decreasingLyapunov}, $\Delta$ is also nonincreasing, thereby showing that $\termset$ is invariant.

Next, we show that the constructed terminal set $\termset$ satisfies Assumption~\ref{as:equilInTermSet}. For any $r \in \mathcal R_\epsilon$, the equilibrium state $\bar{x}_r$ is strictly constraint-admissible by the definition of $\mathcal R_\epsilon$ in~\ref{eq:admissibleRef}, and thus strictly satisfies all linearized constraints by Lemma~\ref{prop:linObsInterior}. As a result, the Lyapunov threshold $\Lambda_i(r)$ in~\ref{eq:lyapunovThreshold} is strictly positive. Consequently, we have
\begin{equation}
    \Delta(\bar x_r, r) = \max_{i\in\mathbb N_{[1,n]}}\! \alpha_i\,[0 - \Lambda_i(r)] < 0.
\end{equation}
Since the obstacle is convex, the projection onto a convex set is Lipschitz \cite[Prop.~4.16]{ref:Bauschke2017}. Therefore, $\Delta$ is continuous because $\Lambda_i$ is continuous for each $i$ and the maximum of continuous functions is continuous. We can thus conclude that $\Delta$ is negative in some neighborhood of $(\bar x_r, r)$. Hence, $(\bar x_r, r) \in \Int\termset$ for all $r \in \mathcal R_\epsilon$.

\bibliographystyle{ieeetr}
\bibliography{IEEEabrv,references}

\vspace*{-\baselineskip}

\begin{IEEEbiography}[{\includegraphics[width=1in, height=1.25in, clip, keepaspectratio]{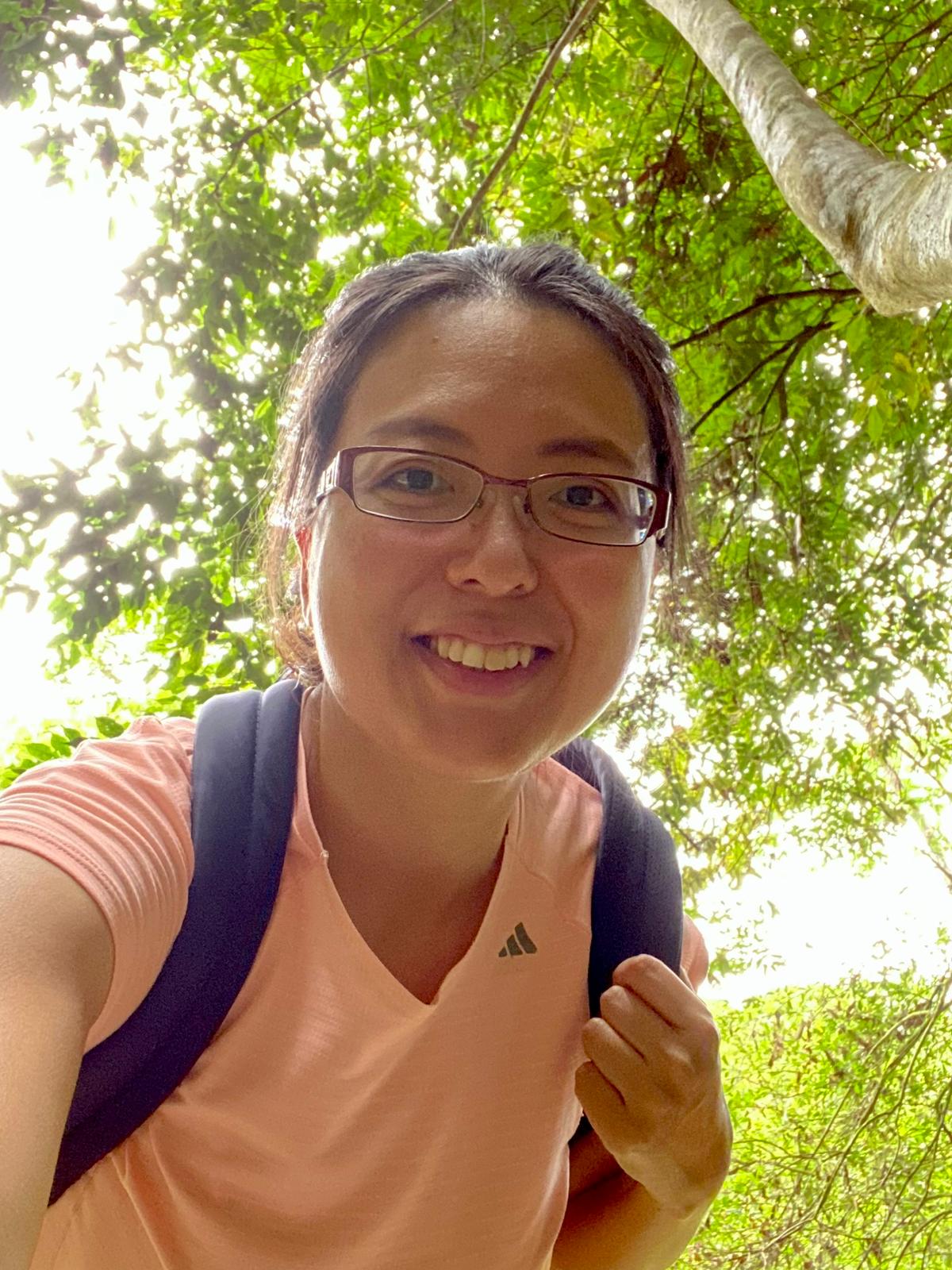}}]
    {Shu Zhang} is a researcher at the Algorithms, Optimization, and Control Lab (AOCL) in the Department of Mechanical Engineering at the University of British Columbia. She received her B.S. from Peking University and her Ph.D. from UBC. Prior to joining the AOCL, she spent six years at Bright Dream Robotics, a company specializing in construction robots, where she led the development of autonomous navigation modules and holds five related patents. Her research focuses on motion planning and multi-agent control, with applications in autonomous aerial and ground robotics.
\end{IEEEbiography}

\vspace*{-2\baselineskip}

\begin{IEEEbiography}[{\includegraphics[width=1in, height=1.33in, trim={0 6cm 0 0}, clip, keepaspectratio]{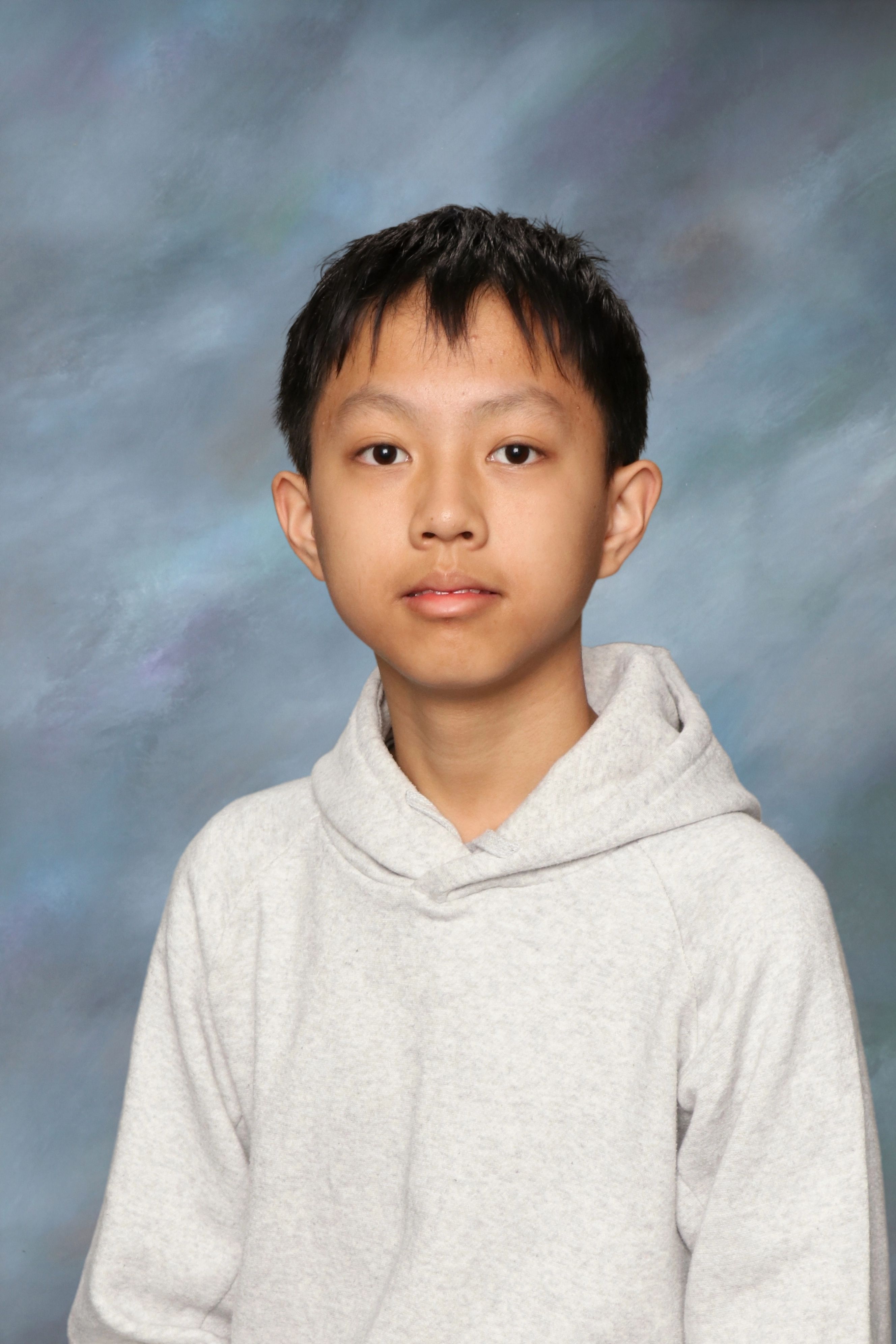}}]
    {James Y. Z. Liu} is a volunteer at the Algorithms, Optimization, and Control Lab at the University of British Columbia. His is a member of the UBC Uncrewed Aircraft Systems Society where he develops control and perception modules for autonomous drones and an active member of the UBC Undergraduate Mathematics Society. Beyond academics, James is a competitive chess player, Rubik's Cube enthusiast, and professional-level origami artist. He also plays drums in the University Hill Music Band and trains as a dragon boat and outrigger paddler.   
\end{IEEEbiography}

\vspace*{-2\baselineskip}

\begin{IEEEbiography}[{\includegraphics[width=1in, height=1.25in, clip, keepaspectratio]{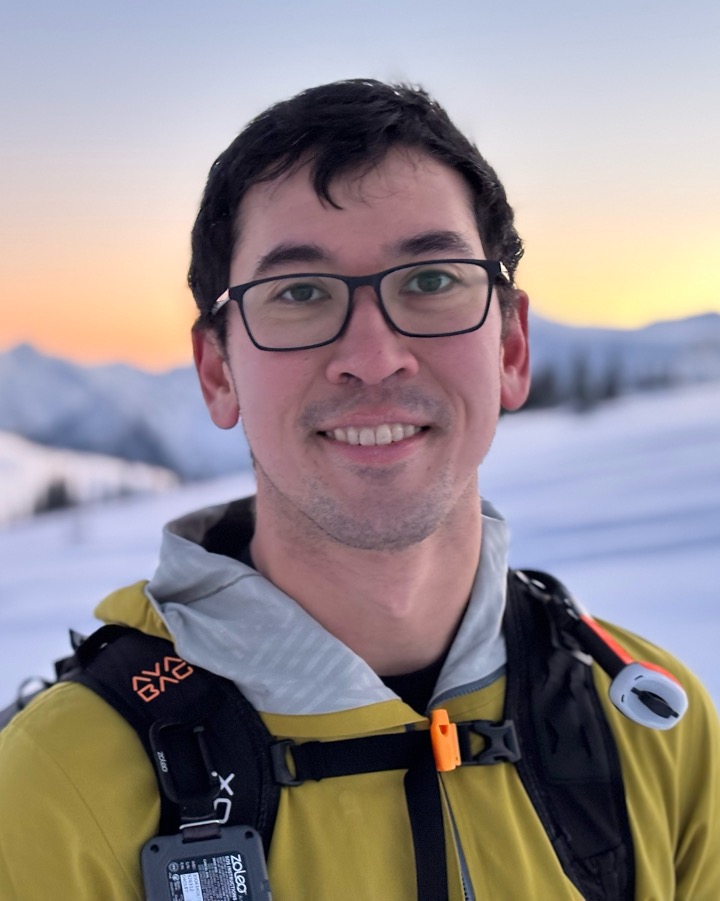}}]
    {Dominic Liao-McPherson} is an Assistant Professor of Mechanical Engineering at the University of British Columbia. He received his BASc. in Engineering Science from the University of Toronto in 2015, his Ph.D. in Aerospace Engineering and Scientific Computing from the University of Michigan, Ann Arbor in 2020. He was then a postdoctoral researcher at the ETH Zürich Automatic Control Laboratory before joining UBC in 2023. His research interests include real-time optimization, algorithmic game-theory, multi-agent control, and edge computing with applications in robotics, energy, and aerospace. 
\end{IEEEbiography}

\vfill

\end{document}